\documentclass[runningheads,a4paper,11pt]{article}
\usepackage[margin=.75in]{geometry}
\usepackage{algorithmicx}
\usepackage{algpseudocode}
\usepackage{bbm} 
\usepackage[numbers,sort&compress]{natbib} 
\usepackage[ruled]{algorithm2e}

\usepackage[T1]{fontenc}
\usepackage[utf8]{inputenc}
\usepackage{url}
\usepackage{subfig}
\usepackage{tabularx, booktabs}
\newcolumntype{Y}{>{\centering\arraybackslash}X}
\usepackage{mathtools}
\usepackage{thmtools}
\usepackage{thm-restate}
\usepackage{graphicx, rawfonts}

\DeclareMathOperator*{\sign}{sign} 

\usepackage[normalem]{ulem}
\usepackage{amssymb,amsmath,graphicx,amsthm}
\usepackage{amsfonts}
\usepackage{color}
\usepackage{placeins}
\newcommand\bbR{\ensuremath{\mathbb{R}}} 
\newcommand\bbE{\ensuremath{\mathbb{E}}} 

\newcommand\bC{\ensuremath{\mathbf{C}}}

\usepackage{enumitem}
\usepackage{balance}

\usepackage{todonotes}
\usepackage{lscape}

\newtheorem{innercustomgeneric}{\customgenericname}
\providecommand{\customgenericname}{}
\newcommand{\newcustomtheorem}[2]{%
	\newenvironment{#1}[1]
	{%
		\renewcommand\customgenericname{#2}%
		\renewcommand\theinnercustomgeneric{##1}%
		\innercustomgeneric
	}
	{\endinnercustomgeneric}
}
\newtheorem{theorem}{Theorem}
\newtheorem{definition}{Definition}

\newtheorem{lemma}{Lemma}

\newcustomtheorem{customthm}{Theorem}
\newcustomtheorem{customprop}{Proposition}
\setlist{nolistsep}
\newcommand\mc[1]{\multicolumn{1}{c}{#1}} 
\newcommand{\subscript}[2]{$#1 _ #2$}

\newcommand\blfootnote[1]{%
	\begingroup
	\renewcommand\thefootnote{}\footnote{#1}%
	\addtocounter{footnote}{-1}%
	\endgroup
}

\begin{document}

\title{Iterative Local Voting for Collective Decision-making in Continuous Spaces}

\author{
	Nikhil Garg\\
	Stanford University\\
	\texttt{nkgarg@stanford.edu} \\
	\and
	Vijay Kamble\\
	University of Illinois at Chicago\\
	\texttt{kamble@uic.edu} \\
	\and
		Ashish Goel\\
	Stanford University\\
	\texttt{ashishg@stanford.edu} \\
	\and
	David Marn\\
	University of California, Berkeley\\
	\texttt{marn@berkeley.edu} \\
		\and
	Kamesh Munagala\\
	Duke University\\
	\texttt{kamesh@cs.duke.edu} \\
}

\maketitle

\begin{abstract}
Many societal decision problems lie in high-dimensional continuous spaces not amenable to the voting techniques common for their discrete or single-dimensional counterparts. These problems are typically discretized before running an election or decided upon through negotiation by representatives. We propose a algorithm called {\sc Iterative Local Voting} for collective decision-making in this setting. In this algorithm, voters are sequentially sampled and asked to modify a candidate solution within some local neighborhood of its current value, as defined by a ball in some chosen norm, with the size of the ball shrinking at a specified rate. 

We first prove the convergence of this algorithm under appropriate choices of neighborhoods to Pareto optimal solutions with desirable fairness properties in certain natural settings: when the voters' utilities can be expressed in terms of some form of distance from their ideal solution, and when these utilities are additively decomposable across dimensions. In many of these cases, we obtain convergence to the societal welfare maximizing solution.

We then describe an experiment in which we test our algorithm for the decision of the U.S. Federal Budget on Mechanical Turk with over 2,000 workers, employing neighborhoods defined by $\mathcal{L}^1, \mathcal{L}^2$ and $\mathcal{L}^\infty$ balls. We make several observations that inform future implementations of such a procedure.\blfootnote{Supported by NSF grant nos. CCF-1408784, CCF-1637397, CCF-1637418, and IIS-1447554, ONR grant no. N00014-15-1-2786, ARO grant no. W911NF-14-1-0526, and the NSF Graduate Research Fellowship under grant no. DGE-114747. This work benefited from many helpful discussions with Oliver Hinder.}
\end{abstract}

\section{Introduction}
Methods and experiments to increase large-scale, direct citizen participation in policy-making have recently become commonplace as an attempt to revitalize democracy. Computational and crowdsourcing techniques involving human-algorithm interaction have been a key driver of this trend~\cite{cabannes_participatory_2004, mcdermott_building_2010, lee_crowdsourcing_2014, quarfoot_quadratic_2017}. Some of the most important collective decisions, whether in government or in business, lie in high-dimensional, continuous spaces -- e.g. budgeting, taxation brackets and rates, collectively bargained wages and benefits, urban planning etc. Direct voting methods originally designed for categorical decisions are typically infeasible for collective decision-making in such spaces. Although there has been some theoretical progress on designing mechanisms for continuous decision-making~\cite{procaccia_approximate_2009,cheng_survey_2015,moulin_strategy-proofness_1980}, in practice these problems are usually resolved using traditional approaches -- they are either discretized before running an election, or are decided upon through negotiation by committee, such as in a standard representative democracy~\cite{cabannes_participatory_2004,shah_participatory_2007,sintomer_porto_2008,gilman_transformative_2012,goel_knapsack_2016}. 

One of the main reasons for the current gap between theory and practice in this domain is the challenge of designing practically implementable mechanisms. We desire procedures that are simple enough to explain and use in practice, and that result in justifiable solutions while being robust to the inevitable deviations from ideal models of user behavior and preferences. To address this challenge, a social planner must first make practically reasonable assumptions on the nature and complexity of feedback that can be elicited from people and then design simple algorithms that operate effectively under these conditions. Further, while robustness to real-world model deviations may be difficult to prove in theory, it can be checked in practice through experiments.

We first tackle the question of what type of feedback voters can give. In general, for the types of problems we wish to solve, a voter cannot fully articulate her utility function. Even if voters in a voting booth had the patience to state their exact utility for a reasonably large number of points (e.g. how much they liked each candidate solution on a scale from one to five), there is no reason to believe that they could do so in any consistent manner. On the other hand, we posit that it is relatively easy for people to choose their favorite amongst a reasonably small set of options, or articulate how they would like to locally modify a candidate solution to better match their preferences. Such an assumption is common and is a central motivation in social choice, especially \textit{implicit utilitarian voting}~\cite{procaccia_distortion_2006}.

In this paper, we study and experimentally test a type of algorithm for large-scale preference aggregation that effectively leverages the possibility of asking voters such easy questions. In this algorithm that we call {\sc Iterative Local Voting} ({\sc ILV}), voters are sequentially sampled and are asked to modify a candidate solution to their favorite point within some local neighborhood, until a stable solution is obtained (if at all). With a continuum of voters, no one votes more than once. The algorithm designer has flexibility in deciding how these local neighborhoods are defined -- in this paper we focus on neighborhoods that are balls in the $\mathcal{L}^q$ norm, and in particular on the cases where $q=1,\,2$ or $\infty$. (For $M < \infty$ dimensional vectors, the $\mathcal{L}^q$ norm $\|x\|_q \triangleq \sqrt[q]{\sum_m |x_m|^q}$. $q=1,\,2$ and $\infty$ neighborhoods correspond to bounds on the sum of absolute values of the changes, the sum of the square of the changes, and the maximum change, respectively.)

More formally, consider a $M$-dimensional societal decision problem in $\mathcal{X} \subset \bbR^M$ and a population of voters $\mathcal{V}$, where each voter $v \in \mathcal{V}$ has bounded utility $f_v(x) \in \bbR, \forall \,x \in \mathcal{X}$. Then we consider the class of algorithms described in Algorithm~\ref{algo:demeq}. We study the algorithm class under two plausible models of how voters respond to query~\eqref{query}, which asks for the voter's favorite point in a local region.
\begin{itemize}[leftmargin=*]
	\item \textbf{Model A: } One possibility is that voters exactly perform the maximization asked of them, responding with their favorite point in the given $\mathcal{L}^q$ norm constraint set. In other words, they return a point ${ \arg\max}_{x\in \{ s : \|s - x_{t-1} \|_q \leq r_t\}} f_{v_t}(x)$. Note that by definition of this movement, the algorithm is \textit{myopically incentive compatible}: if a voter is the last voter and no projections are used, then truthfully performing this movement is the dominant strategy. In general, the mechanism is not globally incentive compatible, nor incentive compatible with projections onto the feasible region. Simple examples of manipulations in both instances exist.
	\item \textbf{Model B: } On the other hand, voters may not actually search within the constraint set to find their favorite point inside of it. Rather, a voter $v$ may have an idea about how to best improve the current point and then move in that direction to the boundary of the given constraint set. This model leads to a voter moving the current solution in the direction of the gradient of her utility function, returning a point $x_{t-1} + r_t \frac{g_{t}}{\| g_{t} \|_q}$, for some $g_{t} \in \partial f_{v_t}(x_{t-1})$. Note that $\partial f(x)$ denotes the set of subgradients of a function $f$ at $x$, i.e. $g\in\partial f(x)$ if $\forall y$, $f(y) - f(x) \geq g^T(y - x)$.
\end{itemize}

{\sc ILV} is directly inspired by the stochastic approximation approach to solve optimization problems \cite{robbins_stochastic_1951}, especially stochastic gradient descent (SGD) and the stochastic subgradient method (SSGM). The idea is that if (a) voter preferences are drawn from some probability distribution and (b) the response of a voter to the query~\eqref{query} moves the solution approximately in the direction of her utility gradient, then this procedure \emph{almost} implements stochastic gradient descent for minimizing negative expected utility. 
\begin{algorithm}[t]
	\SetAlgorithmName{Algorithm}{}{}
	\textit{Inputs: } Initial solution $x_0\in \mathcal{X}$, tolerance $\epsilon>0$, an integer $N$, initial radius $r_0>0$, termination time $T$, norm $q$ for local neighborhood. \\
	\textit{Output: } Solution $x$.
	\begin{itemize}
		\item For $t\geq 1$, sample a voter $v_t\in \mathcal{V}$ at random from the population; set $r_t=r_0/t$ and elicit 
		\begin{equation}\label{query}
		x'_t ={ \arg\max}_{x\in \{ s : \|s - x_{t-1} \|_q \leq r_t\}} f_{v_t}(x),
		\end{equation}
		and then compute $x_t= [x'_t]_\mathcal{X}$, where $[\cdot]_\mathcal{X}$ is a projection onto space $\mathcal{X}$; i.e. ask the voter to move to her favorite point within the constraints, and find the projection of the reported point onto $\mathcal{X}$. 
		\item Stop when either $t=T$, in which case return $x_T$, or when $\max_{l,\,m \in \{t - N,\dots,t\}} |x_l-x_m|\leq \epsilon$, in which case return $x=x_t$.
	\end{itemize}
	\caption{{\sc Iterative Local Voting (ILV)}}
	\label{algo:demeq}
\end{algorithm}

The caveat is that although the procedure can potentially obtain the direction of the gradient of the voter utilities, it cannot in general obtain any information about its magnitude since the movement norm is chosen by the procedure itself. However, we show that for certain plausible utility and voter response models, the algorithm does indeed converge to a unique point with desirable properties, including cases in which it converges to the societal optimum.

Note that with such feedback and without any additional assumptions on voter preferences (e.g. that voter utilities are normalized to the same scale), no algorithm has any hope of finding a desirable solution that depends on the cardinal values of voters' utilities, e.g., the social welfare maximizing solution (the solution that maximizes the sum of agent utilities). This is because an algorithm that uses only ordinal information about voter preferences is insensitive to any scaling or even monotonic transformations of those preferences.

\subsection{Contributions}

This work is a step in extending the vast literature in social choice to continuous spaces, taking into account the feedback that voters can \textit{actually give}. Our main theoretical contributions are as follows:\\
\begin{itemize}[leftmargin=*]
	\item {\bf Convergence for $\mathcal{L}^p$ normed utilities: }We show that if the agents cost functions can be expressed as the $\mathcal{L}^p$ distance from their ideal solution, and if agents correctly respond to query~\eqref{query}, then an interesting duality emerges: for $p=1,\, 2$ or $\infty$, using $\mathcal{L}^q$ neighborhoods, where $q=\infty,\,2$ and $1$ respectively, results in the algorithm converging to the unique social welfare optimizing solution. Whether such a result holds for general $(p,q)$, where $q$ is the dual norm to $p$ (i.e. $1/p + 1/q = 1$), is an open question. However, we show that such a general result holds if, in response to query~\eqref{query}, the voter instead moves the current solution in the direction of the gradient of her utility function to the neighborhood boundary.
	\item  {\bf Convergence for other utilities: }Next, we show convergence to a unique solution in two cases: (a) when the voter cost can be expressed as a weighted sum of $\mathcal{L}^2$ distances over sub-spaces of the solution space, under $\mathcal{L}^2$ neighborhoods -- in which case the solution is also Pareto efficient, and (b)  when the voter utility can be additively decomposed across dimensions, under $\mathcal{L}^\infty$ neighborhoods -- in which case the algorithm converges to the median of the ideal solutions of the voters on each dimension.\\
\end{itemize}

We then build a platform and run the first large-scale experiment in voting in multi-dimensional continuous spaces, in a budget allocation setting. We test three variants of {\sc ILV}: with $\mathcal{L}^1$, $\mathcal{L}^2$ and $\mathcal{L}^\infty$ neighborhoods. Our main findings are as follows:\\
\begin{itemize}[leftmargin=*]
	\item We observe that the algorithm with $\mathcal{L}^\infty$ neighborhoods is the only alternative that satisfies the first-order concern for real-world deployability: consistent convergence to a unique stable solution. Both $\mathcal{L}^1$ and $\mathcal{L}^2$ neighborhoods result in convergence to multiple solutions. 
	\item The consistent convergence under $\mathcal{L}^\infty$ neighborhoods in experiments strongly suggests the decomposability of voter utilities for the budgeting problem. Motivated by this observation, we propose a general class of decomposable utility functions to model user behavior for the budget allocation setting.
	
	\item We make several qualitative observations about user behavior and preferences. For instance, voters have large indifference regions in their utilities, with potentially larger regions in dimensions about which they care about less. Further, we show that asking voters for their ideal budget allocations and how much they care about a given item is fraught with UI biases and should be carefully designed.\\
\end{itemize}

We remark that an additional attractive feature of such a constrained local update algorithm in a large population setting is that strategic behavior from the voters is less of a concern: even if a single voter is strategic, her effect on the outcome is negligible. Further, it may be difficult for a voter, or even a coalition of voters, to strategically vote; one must reason over the possible future trajectories of the algorithm over the randomness of future voters. One coalition strategy for $\mathcal{L}^2$ and $\mathcal{L}^\infty$ neighborhoods, voters trade votes on different dimensions with one another; we leave robustness to such strategies to future work.

The structure of the paper is as follows. After discussing related work in Section~\ref{sec:related}, we present convergence results for our algorithm under different settings in Section~\ref{sec:conv}. In Section~\ref{sec:expdesc}, we introduce the budget allocation problem and describe our experimental platform. In Section~\ref{sec:convergencepatterns}, we analyze the experiment results, and then we conclude the paper in Section~\ref{sec:conclusion}. The proofs of our results are in the appendix.

\section{Related Work}\label{sec:related}
Our work relates to various strands of literature. We note that a conference version of this work appeared previously~\cite{garg_collaborative_2017}. Furthermore, the term ``iterative voting'' is also used in other works to denote unrelated methods~\cite{meir_convergence_2010,airiau_iterated_2009}.
\paragraph{Stochastic Gradient Descent}
As discussed in the introduction, we draw motivation from the stochastic subgradient method (SSGM), and our main proof technique is mapping our algorithm to SSGM. Beginning with the original stochastic approximation algorithm by Robbins and Monro \cite{robbins_stochastic_1951}, a rich literature surrounds SSGM, for instance see~\cite{boyd_subgradient_2006,nemirovski_robust_2009,jiang_scheduling_2010,shor_nondifferentiable_1998}. 
\paragraph{Iterative local voting} A version of our algorithm, with $\mathcal{L}^2$ norm neighborhoods, has been proposed independently several times \cite{hylland_mechanism_1980,chung_directional_2018,benjamin_aggregating_2013} and is referred to as Normalized Gradient Ascent (NGA). Instead of directly asking voters to perform query~\eqref{query}, the movement $ \frac{\nabla f_v(x_{t-1})}{\|\nabla f_v(x_{t-1})\|_2}$ would be estimated through population surveys to try to compute the fixed point where $\textup{E}_v\left[\frac{\nabla f_v(x)}{\|\nabla f_v(x)\|_2}\right] = 0$. (Note that we work with distributions of voters and for strictly concave utility functions, the movement for each voter is well-defined for all but a measure 0 set. Then, given a bounded density function of voters, the expectation is well-defined).

This fixed point has been called Directional Equilibrium (DE) in the recent literature \cite{chung_directional_2018}. The movement is equivalent to the movement in this work in the case voters respond according to Model B and with $\mathcal{L}^2$ neighborhoods, and we show in Section~\ref{sec:equiv} that, in such cases, the algorithm converges to a Directional Equilibrium when it converges. We further conjecture that even under voter Model A, if Algorithm~\ref{alg:votermovetomin} converges, the fixed point is a Directional Equilibrium.

Several properties of the fixed point have been studied, starting from \cite{hylland_mechanism_1980} to more recently, \cite{chung_directional_2018} and \cite{benjamin_aggregating_2013}: it exists under light assumptions, is Pareto efficient, and has important connections to the Majority Core literature in economics. Showing that an iterative algorithm akin to ours converges to such a point has been challenging; indeed, except for special cases such as quadratic utilities $f_v(x) = -(x - x^v)^T\Omega(x - x^v)$, with society-wide $\Omega$ that encodes the relative importance and relationships between issues  \cite{benjamin_aggregating_2013}, convergence is an open question.

Our algorithm differs from NGA in a few crucial directions, even in the case that the movement is equivalent: by relating our algorithm to SGD, we are able to characterize the step-size behavior necessary for convergence and show convergence even when each step is made by a single voter, rather than after an estimate of the societal normalized gradient. One can also characterize the convergence rate of the algorithm~\cite{nemirovski_robust_2009}. Furthermore, the literature has referred to the $\mathcal{L}^2$ norm (or ``quadratic budget'') constraint as ``central to their strategic properties''~\cite{benjamin_relationship_2017}. In this work, this limitation is relaxed -- the same strategic property, myopic incentive compatibility, holds for the other norm constraints for their respective cases.

Finally, because we are primarily interested in designing implementable voting mechanisms, we focus on somewhat different concerns than the directional equilibria literature. However, we believe that the ideas in this work, especially the connections to the optimization literature, may prove useful to work on NGA.  To the best of our knowledge, no work studies such an algorithm with other neighborhoods and under ordinal feedback, or implements such an algorithm. 
\paragraph{Optimization without gradients}
Because we are concerned with optimization without access to voters' utility functions or its gradients, this work seems to be in the same vein as recent literature on convex optimization without gradients -- such as with comparisons or with pairs of function evaluations \cite{flaxman_online_2005,jamieson_query_2012,duchi_dual_2012,duchi_optimal_2015}. However, in the social choice or human optimization setting, we cannot estimate each voter's utility functions or gradients exactly rather than up to a scaling term, and yet we would like to find some point with good societal properties. This limitation prevents the use of strategies from such works.

\citet{jamieson_query_2012}, for example, present an optimal coordinate-descent based algorithm to find the optimum of a function for the case in which noisy comparisons are available on that function; in our setting, such an algorithm could be used to find the optimal value for \textit{each voter}, but not the societal optimum because each voter can independently scale her utility function. \citet{duchi_dual_2012} present a distributed optimization algorithm where each node (voter) has access to its own subgradients and a few of its neighbors, but in our case each voter can arbitrarily scale her utility function and thus her subgradients. Similar problems emerge in applying results from the work of~\citet{duchi_optimal_2015}. In our work, such scaling does not affect the point to which the algorithm converges. 
\paragraph{Participatory Budgeting} The experimental setting for this work, and a driving motivation, is Participatory Budgeting, in which voters are asked to help create a government budget. Participatory budgeting has been among the most successful programs of Crowdsourced Democracy, with deployments throughout the world allocating hundreds of millions of dollars annually, and studies have shown its civic engagement benefits~\cite{cabannes_participatory_2004,shah_participatory_2007,sintomer_porto_2008,gilman_transformative_2012,goel_knapsack_2016,mcdermott_building_2010,lee_crowdsourcing_2014}.

In a typical election, community members propose projects, which are then refined and voted on by either their representatives or the entire community, through some discrete aggregation scheme. In no such real-world election, to our knowledge, can the amount of money to allocate to a project be determined in a continuous space within the voting process, except through negotiation by representatives.

\citet{goel_knapsack_2016} propose a ``Knapsack Voting'' mechanism in which each voter is asked to create a valid budget under the budget constraint; the votes are then aggregated using K-approval aggregation on each dollar in the budget, allowing for fully continuous allocation in the space. This mechanism is strategy-proof under some voter utility models. In comparison, our mechanism works in more general spaces and is potentially easier for voters to do. 

\paragraph{Implicit Utilitarian Voting} With a finite number of candidates, the problem of optimizing some societal utility function (based on the cardinality of voter utilities) given only ordinal feedback is well-studied, with the same motivation as in this work: ordinal feedback such as rankings and subset selections are relatively easy for voters to provide. The focus in such work, referred to as \textit{implicit utilitarian voting}, is to minimize the distortion of the output selected by a given voting rule, over all possible utility functions consistent with the votes, i.e. minimize the worst case error achieved by the algorithm due to an under-determination of utility functions when only using the provided inputs~\cite{procaccia_distortion_2006,caragiannis_voting_2011,goel_metric_2017,caragiannis_subset_2017}. In this work, we show convergence of our algorithm under certain implicit utility function forms. However, we do not characterize the maximum distortion of the resulting fixed point (or even the convergence to any fixed point) under any utility functions consistent with the given feedback, leaving such analysis for future work.  

\section{Convergence Analysis}\label{sec:conv}
In this section, we discuss the convergence properties of {\sc ILV} under various utility and behavior models. For the rest of the technical analysis, we make the following assumptions on our model.

\begin{enumerate}[label=\subscript{\textbf{C}}{\arabic*}]
	\item The solution space $\mathcal{X}\subseteq \mathbb{R}^M$ is non-empty, bounded, closed, and convex.
	\item Each voter $v$ has a unique ideal solution $x_v\in \mathcal{X}$.
	\item The ideal point $x_v$ of each voter is drawn independently from a probability distribution with a bounded and measurable density function $h_\mathcal{X}$. 
\end{enumerate}
Under this model, for a solution $x\in\mathcal{X}$, the societal utility is given by $\textup{E}_v[f_v(x)]$. and the social optimal (SO) solution is any $x^*\in \arg\max_{x\in\mathcal{X}}\textup{E}_v[f_v(x)]$. 

``Convergence'' of {\sc ILV} refers to the convergence of the sequence of random variables $\{x_t\}_{t\geq 1}$ to some $x\in\mathcal{X}$ with probability $1$, assuming that the algorithm is allowed to run indefinitely (this notion of convergence also implies the termination of the algorithm with probability $1$).

In the following subsections, we present several classes of utility functions for which the algorithm converges, summarized in Table~\ref{tab:totalsummary}. We further formalize the relationship to directional equilibria in Section~\ref{sec:equiv}.

\begin{table}
	\centering
	\begin{tabular}{l|m{4cm}| m{4cm}|}
		\mc{} & \mc{\textbf{Model A}} & \mc{\textbf{Model B}} \\		\cline{2-3}
		\textbf{Spatial}, $(p,q) = (2, 2)$, $(1, \infty)$, or $(\infty, 1)$ & \multicolumn{2}{c|}{Social Opt. (Thm \ref{thm:pqmain})} \\\cline{2-3}
		\textbf{Spatial}, $(p,q)\text{ s.t. } 1/p+1/q = 1$ & {$\,\,\,\,\,\,\,\,\,\,\,\,\,\,\,\,\,\,\,\,\,\,\,\,\,\,\,\,\,\,\,\,$?}  &{$\,\,\,\,\,\,\,\,$    Social Opt. (Thm \ref{thm:pqother})} \\\cline{2-3}
		\textbf{Weighted Euclidean} & \multicolumn{2}{c|}{Social Opt. (Thm \ref{thm:other1})} \\\cline{2-3}
		\textbf{Decomposable} & \multicolumn{2}{c|}{Medians (Thm \ref{thm:other2})} \\\cline{2-3}
	\end{tabular}
	\caption{Summary of convergence results}
	\label{tab:totalsummary}
\end{table}

\subsection{Spatial Utilities}
\label{sec:spatialutilities}

Here we consider \emph{spatial} utility functions, where the utilities of each voters can be expressed in the form of some kind of spatial distance from their ideal solutions. First, we consider the following kind of utilities.
\begin{definition}{\bf $\mathcal{L}^p$ normed utilities.}
	The voter utility function is $\mathcal{L}^p$ normed if $f_v(x) = -\|x - x_v\|_p, \forall x\in \mathcal{X}$. 
\end{definition} 

Under such utilities, for $p=1,\,2$ and $\infty$, restricting voters to a ball in the dual norm leads to convergence to the societal optimum.
\begin{restatable}{theorem}{thmpqmain}\label{thm:pqmain}
	Suppose that conditions $\bC_1$, $\bC_2$, and $\bC_3$ are satisfied, the voter utilities are $\mathcal{L}^p$ normed, and voters respond to query $\eqref{query}$ according to either \textbf{Model A} or \textbf{Model B}. Then, {\sc ILV} with $\mathcal{L}^q$ neighborhoods converges to the societal optimal point w.p. $1$ when $(p,q) = (2, 2)$, $(1, \infty)$, or $(\infty, 1)$. 
\end{restatable}
The proof is contained in the appendix. A sketch of the proof is as follows. For the given pairs $(p,\,q)$, we show that, except in certain `bad' regions, the update rule $x_{t+1} = \arg\min_x[ \|x - x_{v_t}\|_p : \|x - x_t\|_q \leq r_t]$ is equivalent to the stochastic subgradient method (SSGM) update rule $x_{t+1} = x_t - r_t g_t$, for some $g_t \in \partial \textup{E}_v[\|x - x_{v_t}\|_p]$, and that the probability of being in a `bad' region decreases fast enough as a function of $r_t$. We then leverage a standard SSGM convergence result to finish the proof. One natural question is whether the result extends to general dual norms $p,q$, where $1/p + 1/q = 1 $. Unfortunately, the update rule is not equivalent to SSGM in general, and we leave the convergence to the societal optimum for general $(p,q)$ as an open question.

Further, note that even if each voter could scale their utility function arbitrarily, the algorithm would converge to the same point.

However, the general result does hold for general dual norms $(p,\,q)$ if one assumes the alternative behavior model. 

\begin{restatable}{theorem}{thmpqother}\label{thm:pqother}
	Suppose that conditions $\bC_1$, $\bC_2$, and $\bC_3$ are satisfied, the voter utilities are $\mathcal{L}^p$ normed, and voters respond to query $\eqref{query}$ according to \textbf{Model B}. Then, {\sc ILV} with $\mathcal{L}^q$ neighborhoods converges to the societal optimal point w.p. $1$ for any $p>0$ and $q>0$ such that $1/p+1/q = 1$.
\end{restatable}

The proof is contained in the appendix. It uses the following property of $\mathcal{L}^p$ normed utilities: the $\mathcal{L}^q$ norm of the gradient of these utilities at any point other than the ideal point is constant. This fact, along with the voter behavior model, allows the algorithm to implicitly capture the magnitude of the gradient of the utilities, and thus a direct mapping to SSGM is obtained. 
Note that the above result holds even if we assume that a voter moves to her ideal point $x_v$ in case it falls within the neighborhood (since, as explained earlier, the probability of sampling such a voter decreases fast enough).

Next, we introduce another general class of utility functions, which we call \textit{Weighted Euclidean utilities}, for which one can obtain convergence to a unique solution. 
\begin{definition}{\bf Weighted Euclidean utilities.} Let the solution space $\mathcal{X}$ be decomposable into $K$ different sub-spaces, so that $x=(x^1,\dots, x^K)$ for each $x\in \mathcal{X}$ (where $\sum_{k=1}^K \textup{dim}(x^k) = M$). Suppose that the utility function of the voter $v$ is
	$$f_v(x)=-\sum_{k=1}^K \frac{w^k_v}{\|w_v\|_2}\|x^k-x^k_v\|_2.$$
	where $w_v$ is a voter-specific weight vector, then the function is a Weighted Euclidean utility function. We further assume that $w_v\in \mathcal{W}\subset \bbR_+^{K}$ and $x_v$ are independently drawn for each voter $v$ from a joint probability distribution with a bounded and measurable density function, with $\mathcal{W}$ nonempty, bounded, closed, and convex.
\end{definition}
This utility function can be interpreted as follows: the decision-making problem is decomposable into $K$ sub-problems, and each voter $v$ has an ideal point $x^k_v$ and a weight $w^k_v$ for each sub-problem $k$, so that the voter's disutility for a solution is the weighted sum of the Euclidean distances to the ideal points in each sub-problems. Such utility functions may emerge in facility location problems, for example, where voters have preferences on the locations of multiple facilities on a map. This utility form is also the one most closely related to the existing literature on Directional Equilibria and Quadratic Voting, in which preferences are linear. To recover the weighted linear preferences case, set $K = M$, with each sub-space of dimension 1. In this case, the following holds:

\begin{restatable}{proposition}{thmotherone}
	Suppose that conditions $\bC_1$, $\bC_2$, and $\bC_3$ are satisfied, the voter utilities are Weighted Euclidean, and voters correctly respond to query $\eqref{query}$ according to either \textbf{Model A} or \textbf{Model B}. Then, {\sc ILV} with $\mathcal{L}^2$ neighborhoods converges with probability $1$ to the societal optimal point. \label{thm:other1}
\end{restatable}

The intuition for the result is as follows: as long as the neighborhood does not contain the ideal point of the sampled voter, the correct response to query~\eqref{query} under weighted Euclidean preferences is to move the solution in the direction of the ideal point to the neighborhood boundary, which, as it turns out, is the same as the direction of the gradient. Thus with radius $r_t$, the effective movement is $\frac{\nabla f_v(x_t)}{||\nabla f_v(x_t)||_2}$. With (normalized) weighted Euclidean utilities, $\|\nabla f_v(x_t)\|_2 = 1$ everywhere. As before, even if the utilities were not normalized (i.e. not divided by $\|w\|_2$), the algorithm would converge to the same point, as if utility functions were normalized.

\subsection{Decomposable Utilities}
\label{sec:linfsubsection}
Next consider the general class of decomposable utilities, motivated by the fact that the algorithm with $\mathcal{L}^\infty$ neighborhoods is of special interest since they are easy for humans to understand: one can change each dimension up to a certain amount, independent of the others. 

\begin{definition}{\bf Decomposable utilities.}
	A voter utility function is decomposable if there exists concave functions $f^m_v$ for $m\in\{1 \dots M\}$ such that  $f_v(x)=\sum_{m=1}^Mf^m_v(x^m)$.
\end{definition}

If the utility functions for the voters are decomposable, then we can show that our algorithm under $\mathcal{L}^\infty$ neighborhoods converges to the vector of medians of voters' ideal points on each dimension. Suppose that $h_{\mathcal{X}}^m$ is the marginal density function of the random variable $x^m_v$, and let $\bar{x}^m$ be the set of medians of $x^m_v$. (By \textit{set} of medians, we mean the set of points such that, on each dimension, the mass of voters with ideal points above and below.)

\begin{restatable}{proposition}{thmothertwo}
	Suppose that conditions $\bC_1$, $\bC_2$, and $\bC_3$ are satisfied, the voter utilities are decomposable, and voters respond to query $\eqref{query}$ according to either \textbf{Model A} or \textbf{Model B}. Then, {\sc ILV} with $\mathcal{L}^\infty$ neighborhoods converges with probability $1$ to a point in the set of medians $\bar{x}$.  \label{thm:other2}
\end{restatable}

Although simply eliciting each agent's optimal solution and computing the vector of median allocations on each dimension is a viable approach in the case of decomposable utilities, deciding an optimal allocation across multiple dimensions is a more challenging cognitive task than deciding whether one wants to increase or decrease each dimension relative to the current solution (see Section~\ref{sec:againstfullelicit_time} for experimental evidence). In fact, in this case, the algorithm can be run separately for each dimension, so that each voter expresses her preferences on only one dimension, drastically reducing the cognitive burden of decision-making on the voter, especially in high dimensional settings like budgeting.

\subsection{Equivalence to Directional Equilibrium}
\label{sec:equiv}

As discussed in Section~\ref{sec:related}, our algorithm, with $\mathcal{L}^2$-norm neighborhoods, is related to an algorithm, NGA, to find what are called Directional Equilibria in literature. Prior work mostly focuses on the properties of the fixed point, with discussion of the proposed algorithm limited to simulations. We show that with the radius decreasing as $\mathcal{O}(\frac{1}{t})$, the algorithm indeed finds directional equilibria in the following sense: if under a few conditions a trajectory of the algorithm converges to a point, then that point is a directional equilibrium. 

\begin{restatable}{theorem}{thmequiv}
	Suppose that $\bC_1$, $\bC_2$, and $\bC_3$ are satisfied, and let $G(x) \triangleq \bbE_v\left[\frac{\nabla f_v(x)}{\| \nabla f_v(x)\|_2}\right]$. Suppose, $G(x)$ is uniformly continuous, $\mathcal{L}^2$ movement norm constraints are used, and voters move according to \textbf{Model B}. If a trajectory $\{x\}_{t=1}^\infty$ of the algorithm converges to $x^*$, i.e. $x_t \to x^*$, then $x^*$ is a directional equilibrium, i.e. $G(x^*) = 0$. 
	\label{thm:equiv}
\end{restatable}

The proof is in the appendix. It relies heavily on the continuity assumption: if a point $x$ is not a directional equilibrium, then the algorithm with step sizes $\mathcal{O}(\frac{1}{t})$ will with probability $1$ leave any small region surrounding $x$: the net drift of the voter movements is away from the region. We note that the necessary assumptions hold for all utility functions for which convergence holds, using the $\mathcal{L}^2$ norm algorithm (e.g. weighted Euclidean utilities). It is further possible to characterize other utility functions for which the equivalence holds: with appropriate conditions on the distribution of voters and how $f$ differs among voters, the conditions on $G$ can be met.

We further conjecture that even under voter Model A, if Algorithm~\ref{alg:votermovetomin} converges, the fixed point is a Directional Equilibrium. Note that as $r_t \to 0$, $f_v(y)$ can be linearly approximated by the first term of the Taylor series expansion around $x$, for $y \in \{ s : ||s - x ||_2 \leq r_t\}$. Then, to maximize $f_v(y)$ in the region, if the region does not contain $x_v$  voter $v$ chooses $y^*$ $s.t.$ $y^* - x \approx r_t\frac{\nabla f(x)}{||\nabla f(x)||_2}$, i.e., the voter moves the solution approximately in the direction of her gradient to the neighborhood boundary. 

A single step of our algorithm with $\mathcal{L}^2$ neighborhoods is similar to Quadratic Voting~\cite{lalley_quadratic_2015,tideman_efficient_2016} for the same reason. Independently of our work,~\citet{benjamin_relationship_2017} formalize the relationship between the Normalized Gradient Ascent mechanism and Quadratic Voting.

\section{Experiments with Budgets}
\label{sec:expdesc}
We built a voting platform and ran a large scale experiment, along with several extensive pilots, on Amazon Mechanical Turk (\url{https://www.mturk.com}). Over 4,000 workers participating in total counting pilots and the final experiment, with over 2,000 workers participating in the final experiment. The design challenges we faced and voter feedback we received provide important lessons for deploying such systems in a real-world setting.

First we present a theoretical model for our setting. We consider a budget allocation problem on $M$ items, where the items may include both expenditures and incomes. One possibility is to define $\mathcal{X}$ as the space of feasible allocations, such as those below a spending limit, and to run the algorithm as defined, with projections. However, in such cases, it may be difficult to theorize about how voters behave; e.g. if voters knew their answers would be projected onto a budget balanced set, they may respond differently.

Rather, we consider an unconstrained budget allocation problem, one in which a voter's utility includes a term for the budget deficit. Let $\mathcal{E} \subseteq \{1 \dots M\},\, \mathcal{I} = \{1 \dots M\}\setminus \mathcal{E}$ be the expenditure and income items, respectively. Then the general \textit{budget utility function} is $f_v(x)= g_v(x) - d(\sum_{e\in \mathcal{E}} x^e -\sum_{i\in \mathcal{I}}x^i)$, where $d$ is an increasing function on the deficit.

For example, suppose a voter's disutility was proportional to the \textit{square} of the budget deficit (she especially dislikes large budget deficits); then, this term adds complex dependencies between the budget items. In general, nothing is known about convergence of Algorithm~\ref{algo:demeq} with such utilities, as the deficit term may add complex dependencies between the dimensions. However, if the voter utility functions are decomposable across the dimensions and $\mathcal{L}^\infty$ neighborhoods used, then the results of Section~\ref{sec:linfsubsection} can be applied. We propose the following class of decomposable utility functions for the budgeting problem, achieved by assuming that the cost for the deficit is linear, and call the class ``decomposable with a linear cost for deficit," or DLCD.

\begin{definition}
	Let $f_v(x)$ be DLCD if 
	$$f_v(x)=\sum_{m=1}^M f^m_v(x^m) - w_v\left(\sum_{e \in \mathcal{E}} x^e -\sum_{i \in \mathcal{I}} x^i\right),$$
	where $f^m_{v}$ is a concave function for each $m$ and $w_v \in \bbR_+$. 
\end{definition}

In the experiments discussed below in the budget setting, {\sc ILV} consistently and robustly converges with $\mathcal{L}^\infty$ norm neighborhoods. Further, it approximately converges to the medians of the optimal solutions (which are elicited independently), as theorized in Section~\ref{sec:linfsubsection}. Such a convergence pattern suggests the validity of the DLCD model, though we do not formally analyze this claim.

\subsection{Experimental Setup}

We asked voters to vote on the U.S. Federal Budget across several of its major categories: National Defense; Healthcare; Transportation, Science, \& Education; and Individual Income Tax (Note that the US Federal Government cannot just decide to set tax receipts to some value. We asked workers to assume tax rates would be increased or decreased at proportional rates in hopes of affecting receipts.)

 This setting was deemed the most likely to be meaningful to the largest cross-section of workers and to yield a diversity of opinion, and we consider budgets a prime application area in general. The specific categories were chosen because they make up a substantial portion of the budget and are among the most-discussed items in American politics. We make no normative claims about running a vote in this setting in reality, and Participatory Budgeting has historically been more successful at a local level.

One major concern was that with no way to validate that a worker actually performed the task (since no or little movement is a valid response if the solution presented to the worker was near her ideal budget), we may not receive high-quality responses. This issue is especially important in our setting because a worker's actions influence the initial solution future workers see. We thus restricted the experiment to workers with a high approval rate and who have completed over $500$ tasks on Mechanical Turk (MTurk). Further, we offered a bonus to workers for justifying their movements well, and more than $80\%$ of workers qualified, suggesting that we also received high-quality movements. The experiment was restricted to Americans to best ensure familiarity with the setting. Turkprime (\url{https://www.turkprime.com}) was used to manage postings and payment.

\subsection{Experimental Parameters}

Our large scale experiment included 2,000 workers and ran over a week in real-time. Participants of any of the pilots were excluded. We tested the $\mathcal{L}^1, \mathcal{L}^2$, and $\mathcal{L}^\infty$ mechanisms, along with a ``full elicitation'' mechanism in which workers reported their ideal values for each item, and a ``weight'' in $[0,10]$ indicating how much they cared about the actual spending in that item being close to their stated value. 

To test repeatability of convergence, each of the constrained mechanisms had three copies, given to three separate groups of people. Each group consisted of two sets with different starting points, with each worker being asked to vote in each set in her assigned group. Each worker only participates as part of one group, and cannot vote multiple times.

 We used a total of three different sets of starting points across the three groups, such that each group shared one set of starting points with each of the other two groups. This setup allowed testing for repeatability across different starting points and observing each worker's behavior at two points. Workers in one group in each constrained mechanism type were also asked to do the full elicitation after submitting their movements for the constrained mechanism, and such workers were paid extra. These copies, along with the full elicitation, resulted in 10 different mechanism instances to which workers could be allocated, each completed by about 200 workers. 

To update the current point, we waited for 10 submissions and then updated the point to their average. This averaging explains the step-like structure in the convergence plots in the next section. The radius was decreased approximately every 60 submissions, $r_t\approxeq \frac{r_0}{\lceil t/60 \rceil}$. The averaging and slow radius decay rate were implemented in response to observing in the pilots that the initial few voters with a large radius had a disproportionately high impact, as there were not enough subsequent voters to recover from large initial movements away from an eventual fixed point (though in theory this would not be a problem given enough voters). We note that the convergence results for stochastic subgradient methods trivially extend to cover these modifications: the average movement over a batch of submissions starting at the same point is still in expectation a subgradient, and the stepped radius decrease still meets the conditions for valid step-sizes.

\subsection{User Experience}
\begin{figure*}[htb!]
	\centering
	\includegraphics[width=6in]{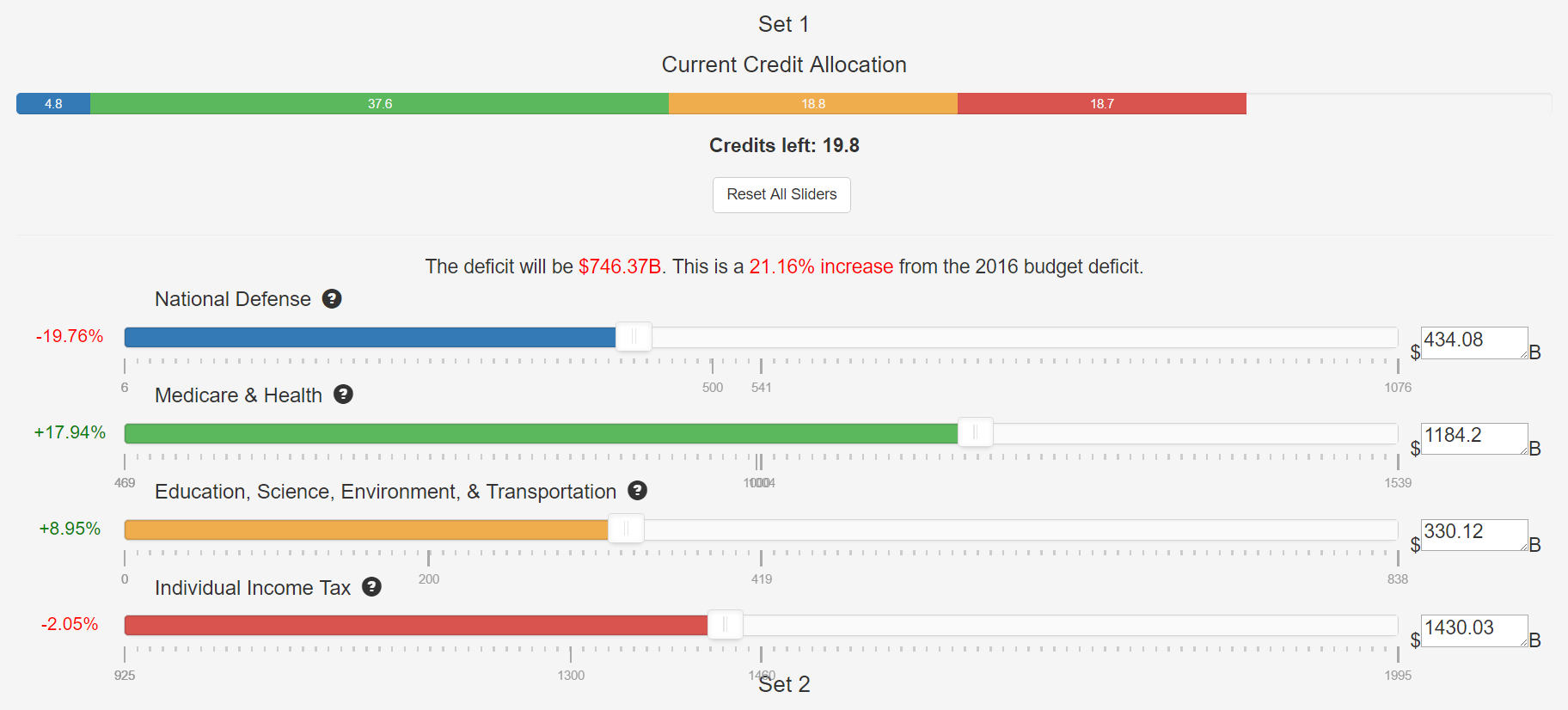}
	\caption{UI Screenshot for 1 set of the $\mathcal{L}^2$ Mechanism}
	\label{fig:screenshots}
\end{figure*}

As workers arrived, they were randomly assigned to a mechanism instance. They had a roughly equal probability of being assigned to each instance, with slight deviations in case an instance was ``busy'' (another user was currently doing the potential $10^\text{th}$ submission before an update of the instance's current point) and to keep the number of workers in each instance balanced. Upon starting, workers were shown mechanism instructions. We showed the instructions on a separate page so as to be able to separately measure the time it takes to read \& understand a given mechanism, and the time it takes to do it, but we repeated the instructions on the actual mechanism page as well for reference.

On the mechanism page, workers were shown the current allocation for each of the two sets in their group. They could then move, through sliders, to their favorite allocation under the movement constraint. We explained the movement constraints in text and also automatically calculated for them the number of ``credits'' their current movements were using, and how many they had left. Next to each budget item, we displayed the percentage difference of the current value from the 2016 baseline federal budget, providing important context to workers (The 2016 budget estimate was obtained from \url{http://federal-budget.insidegov.com/l/119/2016-Estimate} and \url{http://atlas.newamerica.org/education-federal-budget}). We also provided short descriptions of what goes into each budget item as scroll-over text. The resulting budget deficit and its percent change were displayed above the sliders, assuming other budget items are held constant.

For the full elicitation mechanism, workers were asked to move the sliders to their favorite points with no constraints (the sliders went from $\$0$ to twice the 2016 value in that category), and then were asked for their ``weights'' on each budget item, including the deficit. Figure~\ref{fig:screenshots} shows part of the interface for the $\mathcal{L}^2$ mechanism, not including instructions, with similar interfaces for the other constrained mechanisms. The full elicitation mechanism additionally included sliders for items' weights. On the final page, workers were asked for feedback on the experiment. 

A full walk-through of the experiment with screenshots and link to an online demo is available in the Appendix. We plan on posting the data, including feedback. In general, workers seemed to like the experiment, though some complained about the constraints, and others were generally confused. Some expressed excitement about being asked their views in an innovative manner and suggested that everyone could benefit from participating as, at the least, a thought exercise. The feedback and explanations provided by workers were much longer than we anticipated, and they convince us of the procedure's civic engagement benefits.

\section{Results and Analysis}
\label{sec:convergencepatterns}
We now discuss the results of our experiments.
\subsection{Convergence}
\begin{figure}[htb!]
	\centering
	\subfloat[$\mathcal{L}^1$]{\includegraphics[width=5.5in]{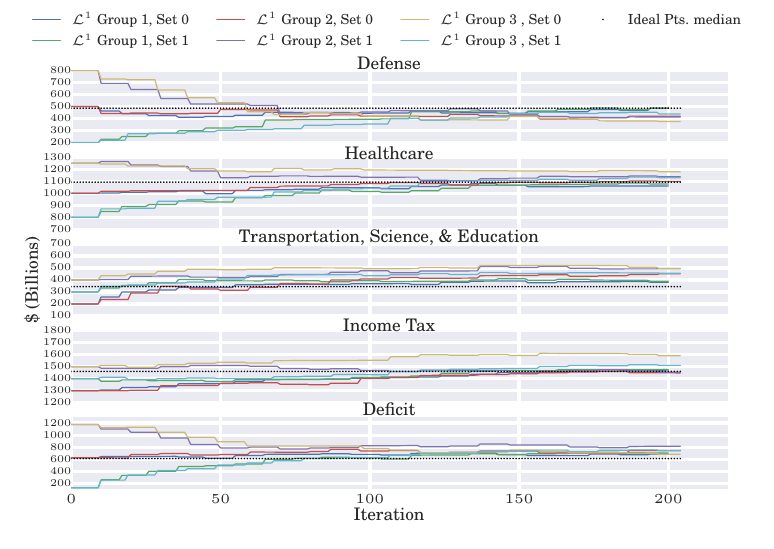}
	}
	\caption{Solution over time for each mechanism type}
\end{figure}
\begin{figure}
	\ContinuedFloat
	\centering
	\subfloat[$\mathcal{L}^2$]{\includegraphics[width=5.5in]{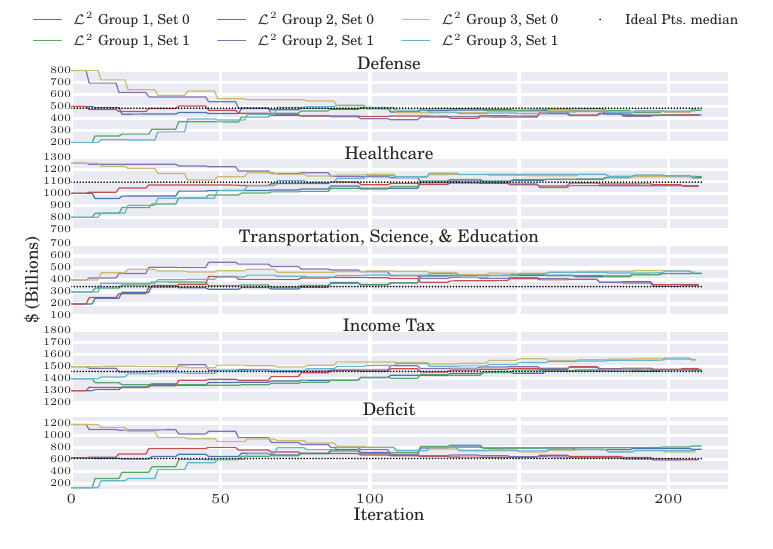}
	}
	
	\subfloat[$\mathcal{L}^\infty$]{\includegraphics[width=5.5in]{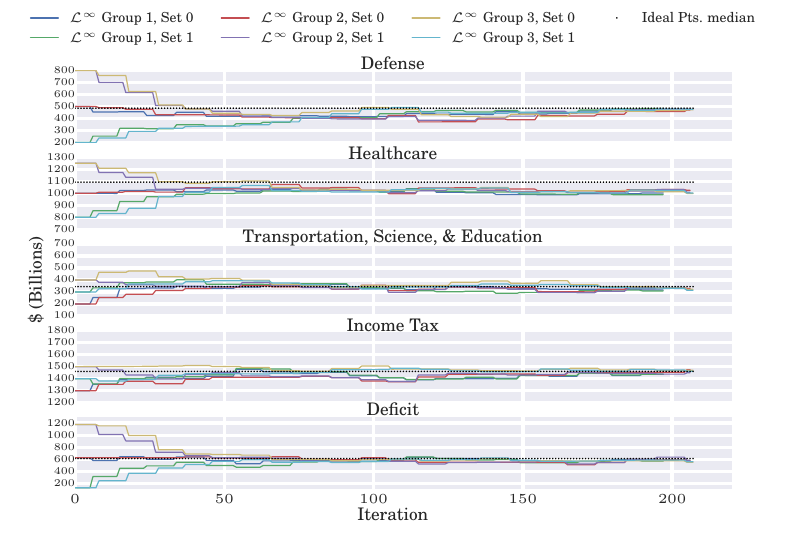}}
	\caption{(Continued) Solution over time for each mechanism type}
	\label{fig:trajectories}
\end{figure}

\begin{figure}[htb!]
	\centering
	\subfloat[$\mathcal{L}^1$]{\includegraphics[width=5.5in]{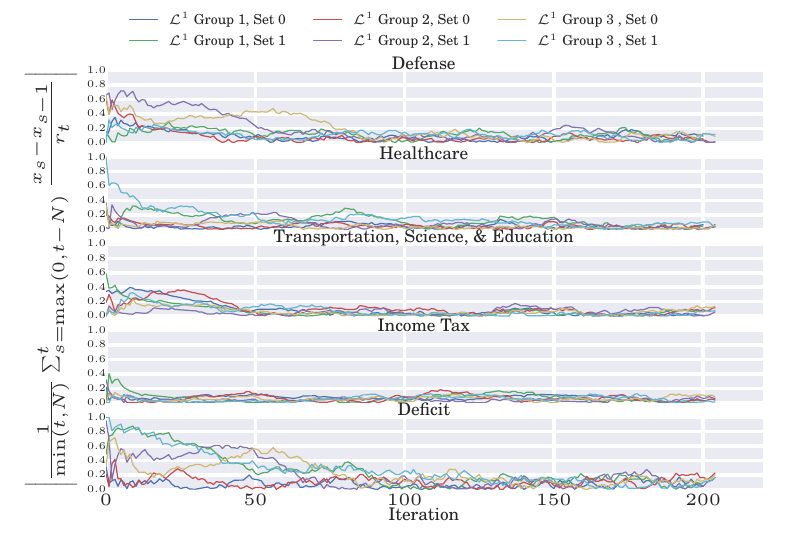}
	}
	\caption{Net normalized movement in window of $N=30$}
\end{figure}
\begin{figure}
	\ContinuedFloat
	\centering
	\subfloat[$\mathcal{L}^2$]{\includegraphics[width=5.5in]{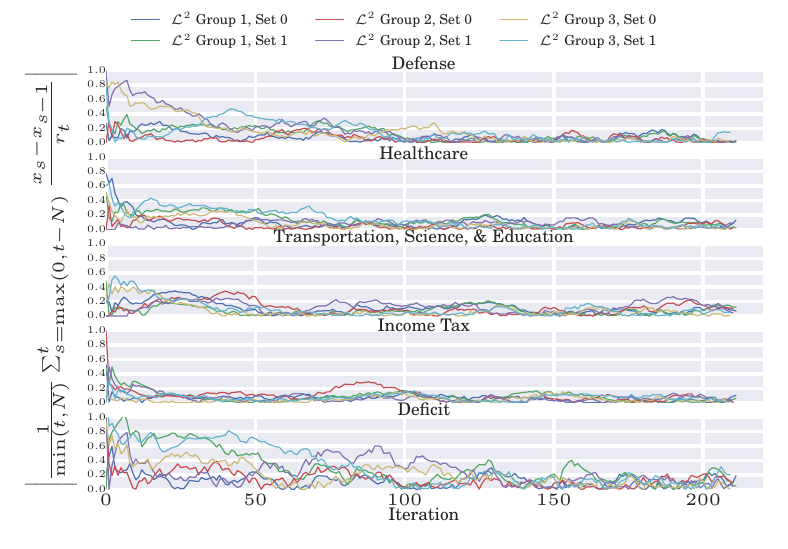}
	}
	
	\subfloat[$\mathcal{L}^\infty$]{\includegraphics[width=5.5in]{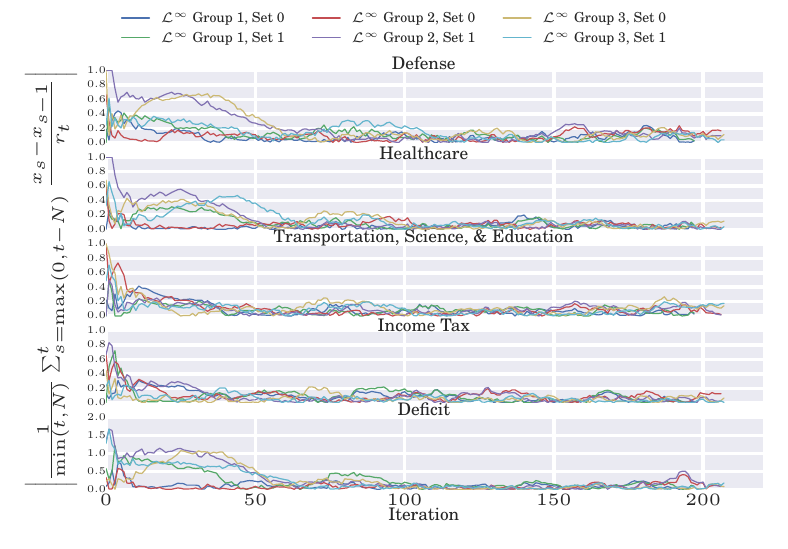}}
	\caption{(Continued) Net normalized movement in window of $N=30$.}
	\label{fig:cumsummovement}
\end{figure}
One basic test of a voting mechanism is whether it produces a consistent and unique solution, given a voting population and their behaviors. If an election process can produce multiple, distinct solutions purely by chance, opponents can assail any particular solution as a fluke and call for a re-vote. The question of whether the mechanisms consistently converge to the same point thus must be answered before analyzing properties of the equilibrium point itself. In this section, we show that the $\mathcal{L}^2$ and $\mathcal{L}^1$ algorithms do not appear to converge to a unique point, while the $\mathcal{L}^\infty$ mechanism converges to a unique point across several initial points and with distinct worker populations.

The solutions after each voter for each set of starting points, across the 3 separate groups of people for each constrained mechanism are shown in Figure~\ref{fig:trajectories}. Each plot shows all the trajectories with the given mechanism type, along with the median of the ideal points elicited from the separate voters who only performed the full elicitation mechanism. Observe that the three mechanisms have remarkably different convergence patterns. In the $\mathcal{L}^1$ mechanism, not even the sets done by the same group of voters (in the same order) converged in all cases. In some cases, they converged for some budget items but then diverged again. In the $\mathcal{L}^2$ mechanism, sets done by the same voters starting from separate starting points appear to converge, but the three groups of voters seem to have settled at two separate equilibria in each dimension. Under the $\mathcal{L}^\infty$ neighborhood, on the other hand, all six trajectories, performed by three groups of people, converged to the same allocation very quickly and remained together throughout the course of the experiment. Furthermore, the final points, in all dimensions except Healthcare, correspond almost exactly to the median of values elicited from the separate set of voters who did only the full elicitation mechanism. For Healthcare, the discrepancy could result from biases in full elicitation (see Section~\ref{sec:againstfullelicit_biases}), though we make no definitive claims. These patterns shed initial insight on how the use of $\mathcal{L}^2$ constraints may differ from theory in prior literature and offer justification for the use of DLCD utility models and the $\mathcal{L}^\infty$ constrained mechanism.

One natural question is whether these mechanisms really have converged, or whether if we let the experiment continue, the results would change. This question is especially salient for the $\mathcal{L}^2$ trajectories, where trajectories within a group of people converged to the same point, but trajectories between groups did not. Such a pattern could suggest that our results are a consequence of the radius decreasing too quickly over time, or that the groups had, by chance, different distributions of voters which would have been corrected with more voters. However, we argue that such does not seem to be the case, and that the mechanism truly found different equilibria. 

We can test whether the final points for each trajectory are stable by checking the net movement in a window, normalized by each voter's radius, i.e. $\frac{1}{N}\sum_{s=t-N}^t \frac{x_s-x_{s-1}}{r_s}$, for some $N$. If voters in a window are canceling each other's movements, then this value would go to 0, and the algorithm would be stable even if the radius does not decrease. The notion is thus robust to apparent convergence just due to decreasing radii.  The net movement normalized in a sliding window of 30 voters, for each dimension and mechanism, is shown in Figure~\ref{fig:cumsummovement}. It seems to almost die down for almost all mechanisms and budget items, except for a few cases which do not change the result. We conclude it likely that the mechanisms have settled into equilibria which are unlikely to change given more voters. 
\FloatBarrier
\subsection{Understanding Voter Behavior}
\label{sec:behavioranalysis}
A mechanism's practical impact depends on more than whether it consistently converges, however. We now turn our attention to understanding how voters behave under each mechanism and whether we can learn anything about their utility functions from that behavior. We find that voters understood the mechanisms but that their behaviors suggest large indifference regions, and that the full elicitation scheme is susceptible to biases that can skew the results.

\subsubsection{Voter understanding of mechanisms}

One important question is whether, given very little instruction on how to behave, voters understand the mechanisms and act approximately optimally under their (unknown to us) utility function. This section shows that the voters behavior follows what one would expect in one important respect: how much of one's movement budget the voter used on each dimension, given the constraint type. 

Regardless of the exact form of the utility function, one would expect that, in the $\mathcal{L}^1$ constrained mechanism, a voter would use most of her movement credits in the dimension about which she cares most. In fact, in either the Weighted Euclidean preferences case (and with `sub-space' being a single dimension) or with a small radius with $\mathcal{L}^1$ constraints, a voter would move only on one dimension. With $\mathcal{L}^2$ constraints, one would expect a voter to apportion her movement more equally because she pays an increasing marginal cost to move more in one dimension (people were explicitly informed of this consequence in the instructions). Under the Weighted Euclidean preferences model with $\mathcal{L}^2$ constraints, a voter would move in each dimension proportional to her weight in that dimension. Finally, with $\mathcal{L}^\infty$ constraints, a voter would move, in all dimensions in which she is not indifferent, to her favorite point in the neighborhood for that dimension (most likely an endpoint), independently of other dimensions. One would thus expect a more equal distribution of movements. 

Figure~\ref{fig:experiment2percentmovement_allmechanisms} shows the average movement (as a fraction of the voter's total movement) by each voter for the dimension she moved most, second, third, and fourth, respectively, for each constrained mechanism. We reserve discussion of the full elicitation weights for Section~\ref{sec:againstfullelicit_biases}. The movement patterns indicate that voters understood the constraints and moved accordingly -- with more equal movements across dimensions in $\mathcal{L}^2$ than in $\mathcal{L}^1$, and more equal movements still in $\mathcal{L}^\infty$. We dig deeper into user utility functions next, but can conclude that, regardless of their exact utility functions, voters responded to the constraint sets appropriately.  

\begin{figure}[htb!]
	\centering
	\includegraphics[width=4.5in]{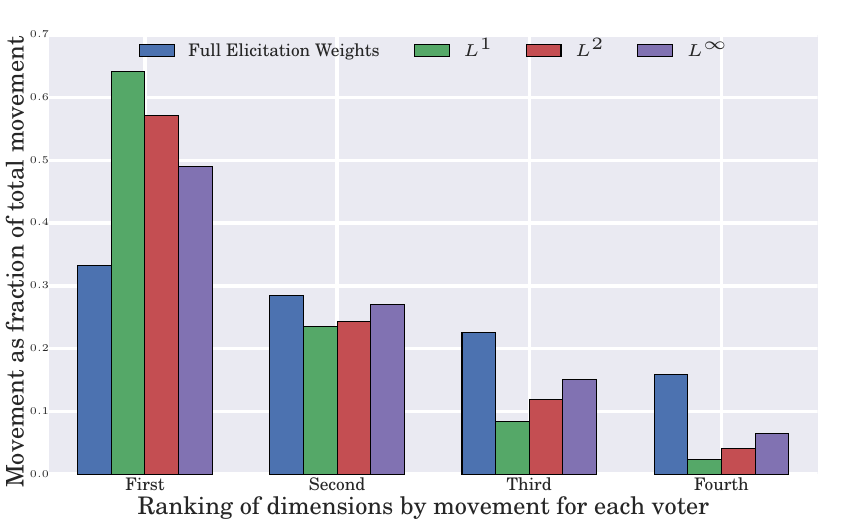}
	\caption{Average movement in dimension over total movement for each voter, with dimensions sorted}
	\label{fig:experiment2percentmovement_allmechanisms}
\end{figure}

\subsubsection{Large indifference regions}
\label{sec:largeindifference}
Although it is difficult to extract a voter's full utility function from their movements, the separability of dimensions (except through the deficit term) under the $\mathcal{L}^\infty$ constraint allows us to test whether voters behave according to some given utility model in that dimension, without worrying about the dependency on other dimensions. 


Figure~\ref{fig:histogramofcreditslinf} shows, for the $\mathcal{L}^\infty$ mechanism, a histogram of the movement on a dimension as a fraction of the radius (we find no difference between dimensions here). Note that a large percentage of voters moved very little on a dimension, even in cases where their ideal point in that dimension was far away (defined as being unreachable under the current radius). This result cannot be explained away by workers clicking through without performing the task: almost all workers moved at least one dimension, and, given that a worker moved in a given dimension, it would not explain smaller movements being more common than larger movements. That this pattern occurs in the $\mathcal{L}^\infty$ mechanism is key -- if a voter feels \textit{any} marginal disutility in a dimension, she can move the allocation without paying a cost of more limited movement in other dimensions. We conclude that, though voters may share a single ideal point for a dimension when asked for it, they are in fact relatively indifferent over a potentially large region -- and their actions reflect so.

We further analyze this claim in Appendix Section~\ref{sec:indifferenceregion}, looking at the same distribution of movement but focusing on workers who provided a text explanation longer than (and shorter than, separately) the median explanation of 197 characters. (We assume that the voters who invested time in providing a more thorough explanation than the average worker also invested time in moving the sliders to a satisfactory point, though this assumption cannot be validated.) Though there are some differences (those who provide longer explanations also tend to use more of their movement), the general pattern remains the same; only about $40\%$ of workers who provided a long explanation and were far away from their ideal point on a dimension used the full movement budget. This pattern suggests that voters are relatively indifferent over large regions. 

Furthermore, this lack of movement is correlated with a voter's weights when she was also asked to do the full elicitation mechanism. Conditioned on being far from her ideal point, when a voter ranked an item as one of her top two important items (not counting the deficit term), she moved an average of $74\%$ of her allowed movement in that dimension; when she ranked an item as one of least two important items, she moved an average of $61\%$, and the difference is significant through a two sample t-test with $p = .013$. We find no significant difference in movement within the top two ranked items or within the bottom two ranked items. This connection suggests that one can potentially determine which dimensions a voter cares about by observing these indifference regions and movements, even in the $\mathcal{L}^\infty$ constrained case. One caveat is that the differences in effects are not large, and so at the individual level inference of how much an individual cares about one dimension over another may be noisy. On the aggregate, however, such determination may prove useful.

Furthermore, we note that while such indifference regions conflict with the utility models under which the $\mathcal{L}^2$ constraint mechanism converges in theory, it fits within the DLCD framework introduced in Section~\ref{sec:expdesc}.

\begin{figure}[htb!]
	\centering
	\includegraphics[width=4.5in]{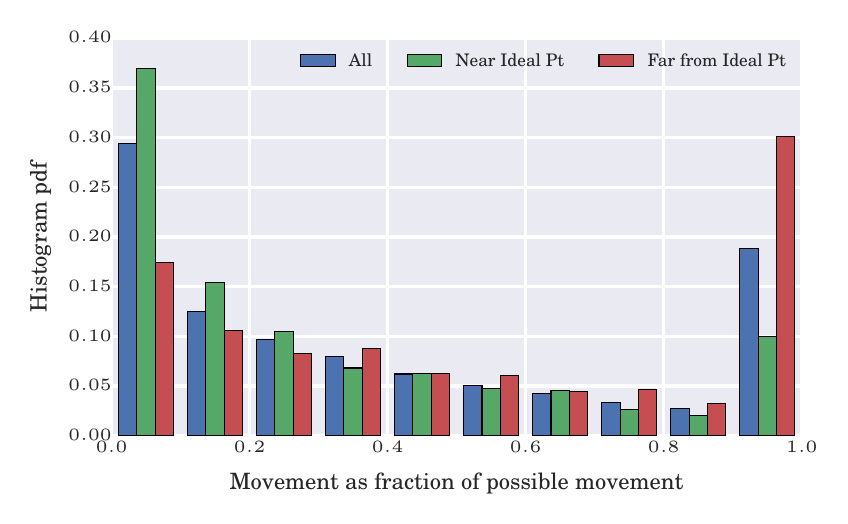}
	\caption{Fraction of possible movement in each dimension in $\mathcal{L}^\infty$, conditioned on distance to ideal pt. The `All' condition contains data from all three $\mathcal{L}^\infty$ instances, whereas the others only from the instance that also did full elicitation.}
	\label{fig:histogramofcreditslinf}
\end{figure}

\FloatBarrier

\subsubsection{Mechanism time}
\label{sec:againstfullelicit_time}
\begin{figure}[htb!]
	\centering
	\includegraphics[width=4.5in]{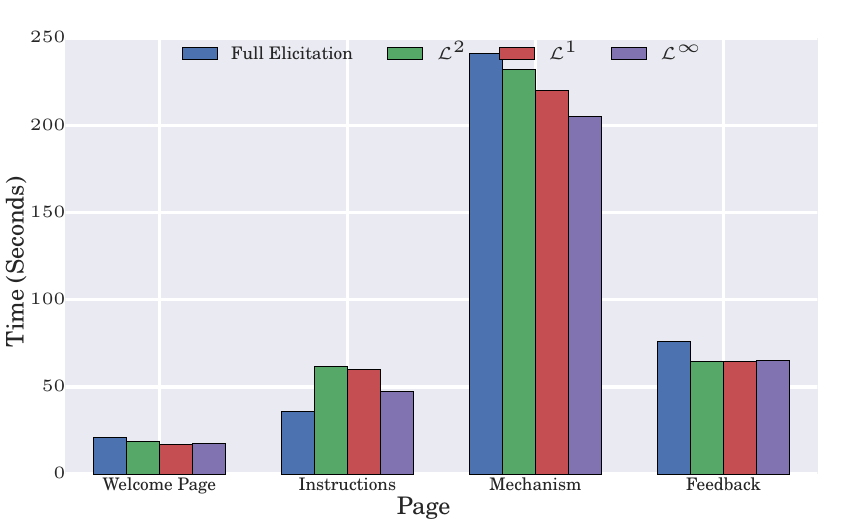}
	\caption{Median time per page}
	\label{fig:pagetimes}
\end{figure}

In this section, we note one potential problem with schemes that explicitly elicit voter's optimal solutions -- for instance, to find the component-wise median -- as compared to the constrained elicitation used in {\sc ILV}: it seems to be cognitively difficult for voters. In Figure~\ref{fig:pagetimes}, the median time per page, aggregated across each mechanism type, is shown. The ``Mechanism'' time includes a single user completing both sets in each of the constrained mechanism types, but not does include the time to also do the extra full elicitation task in cases where a voter was asked to do both a constrained mechanism and the full elicitation. The full elicitation bars include only voters who did only the full elicitation mechanism, and so the bars are completely independent. On average, it took longer to do the full elicitation mechanism than it took to do \textit{two} sets of any of the constrained mechanisms, suggesting some level of cognitive difficulty in articulating one's ideal points and weights on each dimension -- even though understanding what the instructions are asking was simple, as demonstrated by the shorter instruction reading time for the full elicitation mechanism. The $\mathcal{L}^\infty$ mechanism took the least time to both understand and do, while the $\mathcal{L}^2$ mechanism took the longest to do, among the constrained mechanisms. 

This result is intuitive: it is easier to move each budget item independently when the maximum movement is bounded than it is to move the items when the sum or the sum of the changes squared is bounded (even when these values are calculated for the voter). In practice, with potentially tens of items on which constituents are voting, these relative time differences would grow even larger, potentially rendering full elicitation or $\mathcal{L}^2$ constraints unpalatable to voters.

One potential caveat to this finding is that the Full Elicitation mechanism potentially provides more information than do the other mechanisms. From a polling perspective, it is true that more information is provided from full elicitation -- one can see the distribution of votes, the disagreement, and correlation across issues, among other things. However, from a voting perspective, in which the aggregation (winner) is the only thing reported, it is not clear that this extra information is useful. Further, much of this information that full elicitation provides can reasonably be extracted from movements of voters, especially the movements of those who are given a starting point close to the eventual equilibrium.
\FloatBarrier

\subsubsection{UI biases}
\label{sec:againstfullelicit_biases}
We now turn our attention to the question of how workers behaved under the full elicitation mechanism and highlight some potential problems that may affect results in real deployments. Figures~\ref{fig:histofvalues}~and~\ref{fig:histofweights} show the histogram of values and weights, respectively, elicited from all workers who did the full elicitation mechanism. Note that in the histogram of values, in every dimension, the largest peak is at the slider's default value (at the 2016 estimated budget), and the histograms seem to undergo a phase shift at that peak, suggesting that voters are strongly anchored at the slider's starting value. This anchoring could systematically bias the medians of the elicited values. 

A similar effect occurs in eliciting voter weights on each dimension. Observe that in Figure~\ref{fig:experiment2percentmovement_allmechanisms} the full elicitation weights appear far more balanced than the weights implied by any of the mechanisms (for the full elicitation mechanism, the plot shows the average weight over the sum of the weights for each voter). From the histogram of full elicitation weights, however, we see that this result is a consequence of voters rarely moving a dimension's weight down from the default of 5, but rather moving others up.

One potential cause of this behavior is that voters might think that putting high weights on each dimension would mean their opinions would count more, whereas in any aggregation one would either ignore the weights (calculate the unweighted median) or normalize the weights before aggregating. In future work, one potential fix could be to add a ``normalize'' button for the weights, which would re-normalize the weights, or to automatically normalize the weights as voters move the sliders.

These patterns demonstrate the difficulty in eliciting utilities from voters directly; even asking voters how much they care about a particular budget item is extremely susceptible to the user interface design. Though such anchoring to the slider default undoubtedly also occurs in the $\mathcal{L}^\infty$ constrained mechanism, it would only slow the rate of convergence, assuming the anchoring affects different voters similarly. These biases can potentially be overcome by changing the UI design, such as by providing no default value through sliders. Such design choices must be carefully thought through before deploying real systems, as they can have serious consequences.

\begin{figure}[htb!]
	\centering
	\includegraphics[width=4.5in]{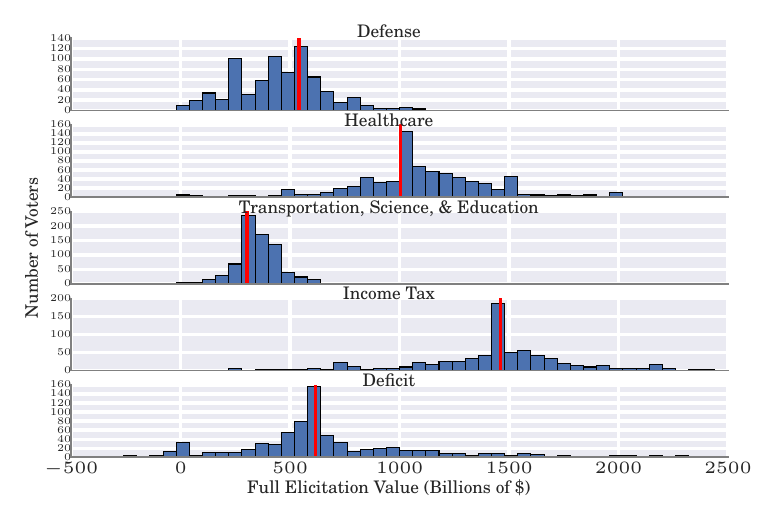}
	\caption{Histogram of values from all full elicitation data. The red vertical lines indicate each slider's default value (at the 2016 estimated budget).}
	\label{fig:histofvalues}
\end{figure}
\begin{figure}[htb!]
	\centering
	\includegraphics[width=4.5in]{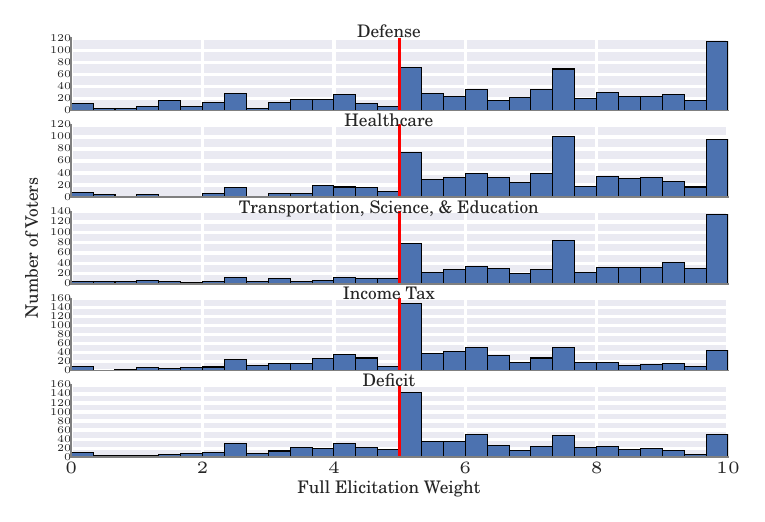}
	\caption{Histogram of weights from all full elicitation data. The red vertical lines indicate the sliders' default value of 5.}
	\label{fig:histofweights}
\end{figure}

\section{Conclusion}\label{sec:conclusion}
We evaluate a natural class of iterative algorithms for collective decision-making in continuous spaces that makes practically reasonable assumptions on the nature of human feedback. We first introduce several cases in which the algorithm converges to the societal optimum point, and others in which the algorithm converges to other interesting solutions. We then experimentally test such algorithms in the first work to deploy such a scheme. Our findings are significant: even with theoretical backing, two variants fail the basic test of being able to give a consistent decision across multiple trials with the same set of voters. On the other hand, a variant that uses $\mathcal{L}^\infty$ neighborhoods consistently leads to convergence to the same solution, which has attractive properties under a likely model for voter preferences suggested by this convergence. We also make certain observations about other properties of user preferences -- most saliently, that they have large indifferences on dimensions about which they care less.

In general, this work takes a significant step within the broad research agenda of understanding the fundamental limitations on the quality of societal outcomes posed by the constraints of human feedback, and in designing innovative mechanisms that leverage this feedback optimally to obtain the best achievable outcomes.

\appendix
\section{Mechanical Turk Experiment Additional Information}

In this section, we provide additional information regarding our Amazon Mechanical Turk experiment, including a walk-through of the user experience. Furthermore, we have a live demo that can be accessed at: \url{http://gargnikhil.com/projectdetails/IterativeLocalVoting/}. This demo will remain online for the foreseeable future.

\noindent Figures~\ref{fig:page1} through~\ref{fig:page3full} show screenshots of the experiment. We now walkthrough the experiment:
\begin{itemize}
	\item Figure~\ref{fig:page1} -- Welcome page. Arriving from Amazon Mechanical Turk, the workers read an introduction and the consent agreement.
	\item Figure~\ref{fig:page2l2} -- Instructions (shown are $\mathcal{L}^2$ instructions). The workers read the instructions, which are also provided on the mechanism page. There is a 5 minute limit for this page.
	\item Figures~\ref{fig:page3l2},~\ref{fig:page3full} -- Mechanism page for $\mathcal{L}^2$ and Full Elicitation, respectively. For the former, workers are asked to move to their favorite point within a constraint set, for 2 different budget points. The ``Current Credit Allocation'' encodes the constraint set -- as workers move the budget bars, it shows how much of their movement budget they have spent, and on which items. The other constrained movement mechanisms are similar. For the Full Elicitation mechanism, voters are simply asked to indicate their favorite budget point and weights. The instructions are repeated on the mechanism page as well at the top. There is a 10 minute limit for this page.
	\item Second mechanism, 30$\%$ of workers. Some workers were asked to do both one of the $\mathcal{L}^1$,$\mathcal{L}^2$, or $\mathcal{L}^\infty$, and the Full Elicitation mechanism. For these workers, the Full Elicitation mechanism shows up after the constrained mechanism.
	\item Figure~\ref{fig:page4} -- Feedback page. Finally, workers are asked to provide feedback, after which they are shown a code and return to the Mechanical Turk website. 
\end{itemize}

\begin{figure}
	\centering
	\begin{subfloat}{}
		\includegraphics[width=1\textwidth]{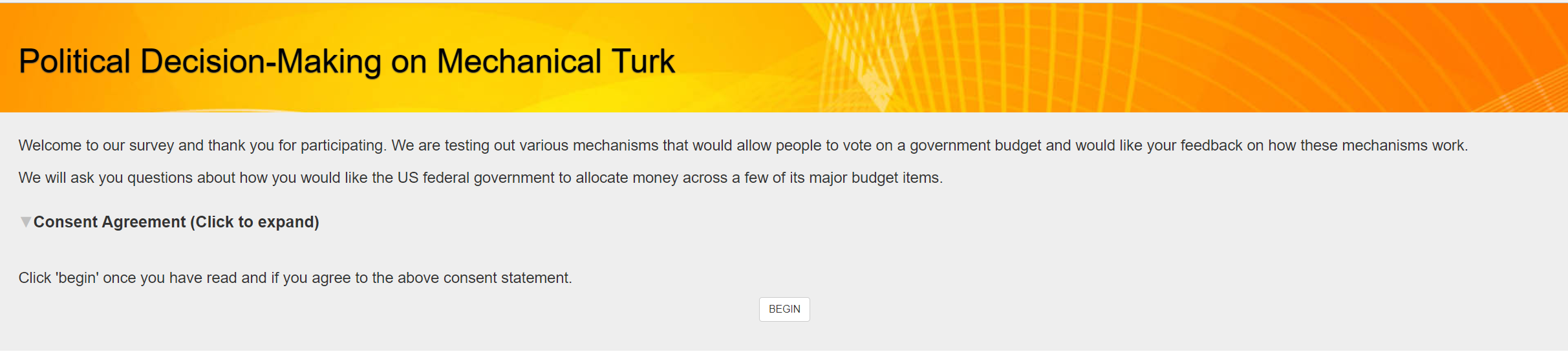}
		\caption{Page 1 -- Welcome Page for all mechanisms}
		\label{fig:page1}
	\end{subfloat}
\end{figure}
\begin{figure}
	\begin{subfloat}{}
		\includegraphics[width=1\textwidth]{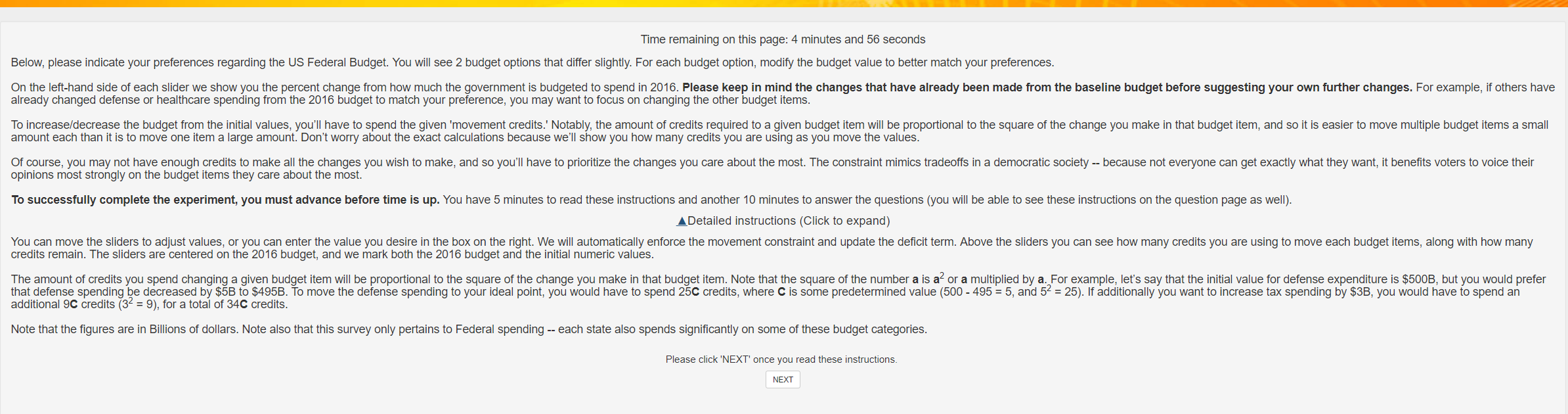}
		\caption{$\mathcal{L}^2$ Page 2 -- Instructions}
		\label{fig:page2l2}
	\end{subfloat}
\end{figure}
\begin{figure}
	\begin{subfloat}{}
		\includegraphics[width=1\textwidth]{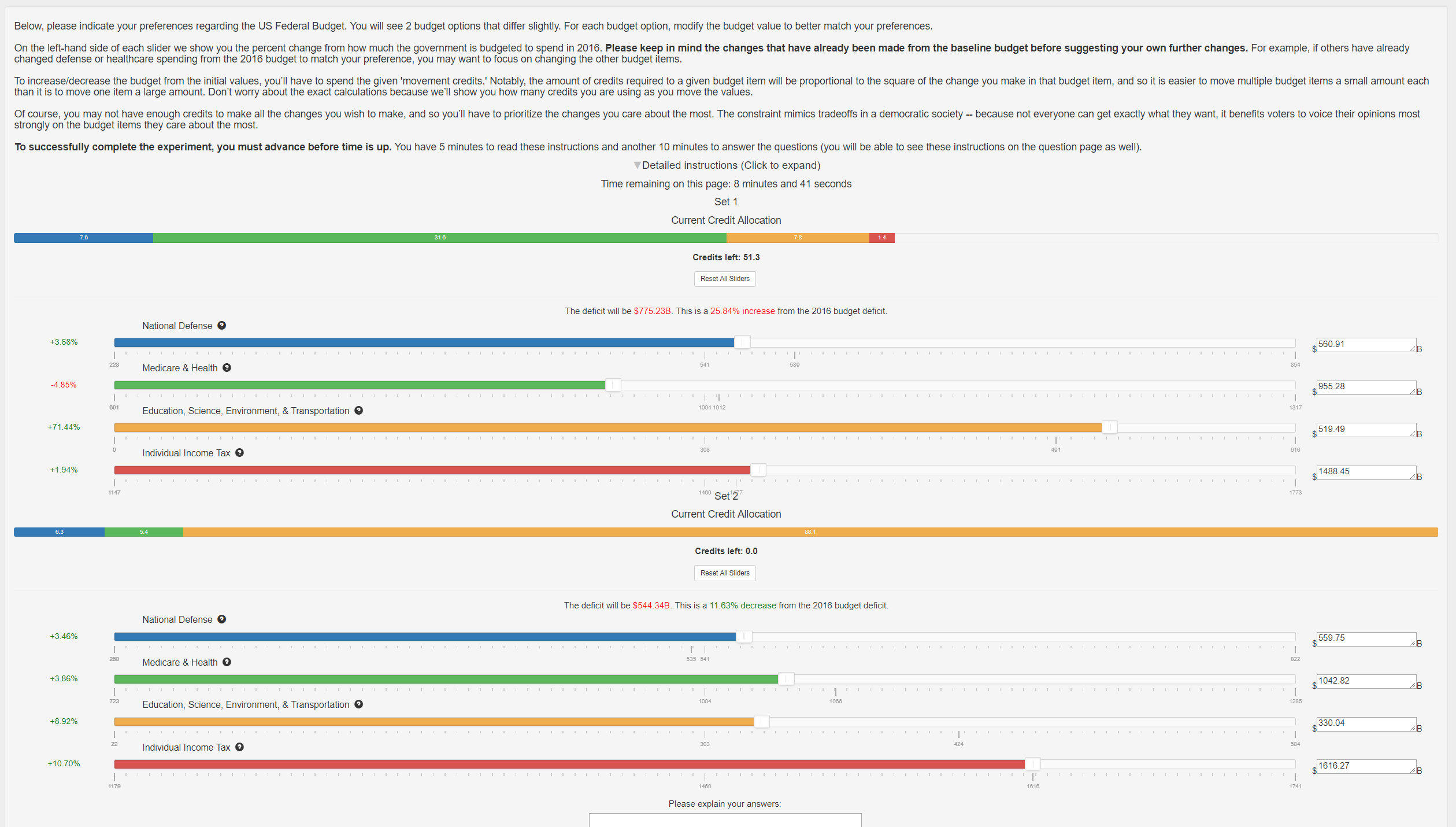}
		\caption{$\mathcal{L}^2$ Page 3 -- Mechanism}
		\label{fig:page3l2}
	\end{subfloat}
\end{figure}
\begin{figure}
	\begin{subfloat}{}
		\includegraphics[width=1\textwidth]{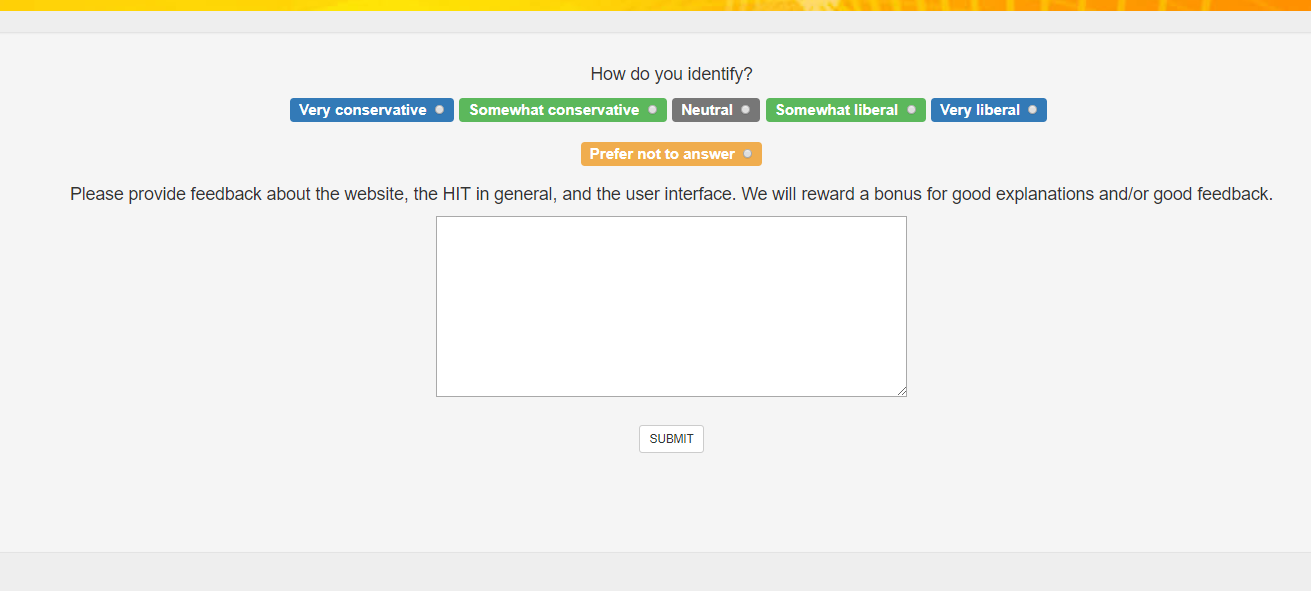}
		\caption{Page 4 -- Feedback for all mechanisms}
		\label{fig:page4}
	\end{subfloat}
\end{figure}
\begin{figure}
	\begin{subfloat}{}
		\includegraphics[width=1\textwidth]{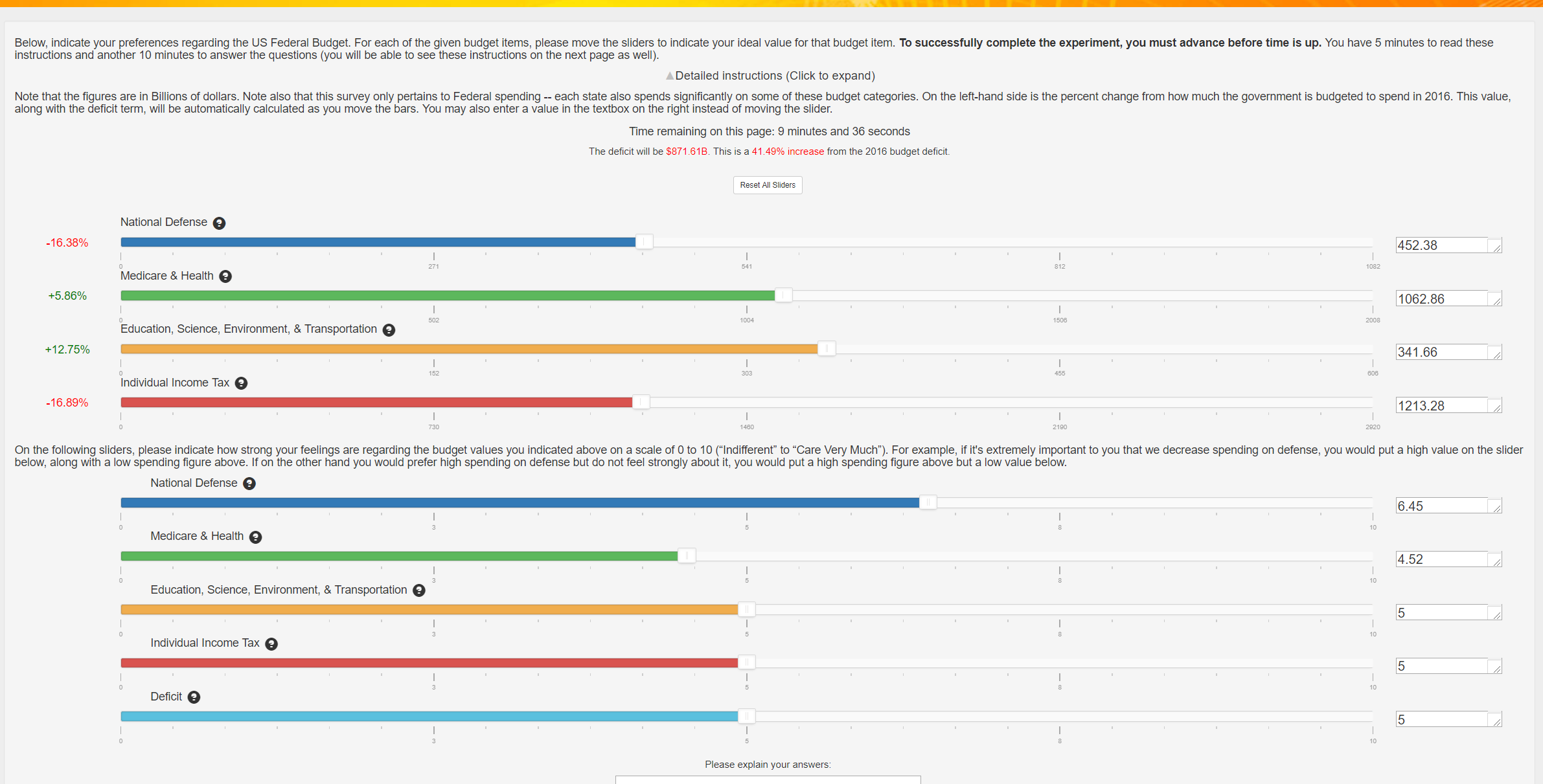}
		\caption{Full Elicitation Page 3 -- Mechanism}
		\label{fig:page3full}
	\end{subfloat}
\end{figure}

\section{Indifference Regions Additional Information}
\label{sec:indifferenceregion}
We now present some additional data for the claim in Section~\ref{sec:largeindifference}, that voters have large indifference regions on the space. In particular, Figures~\ref{fig:histogramofcreditslinf_cared} and~\ref{fig:histogramofcreditslinf_caredNOT} reproduce Figure~\ref{fig:histogramofcreditslinf} but with workers who provided explanations longer (and shorter) than the median response, respectively. This split can (roughly) correspond to workers who may have answered more or less sincerely to the budgeting question. We find that the response distribution, as measured by the fraction of possible movement one used when far away from one's ideal point on a given dimension, are similar. 

\begin{figure}[htb!]
	\centering
	\includegraphics[width=4.5in]{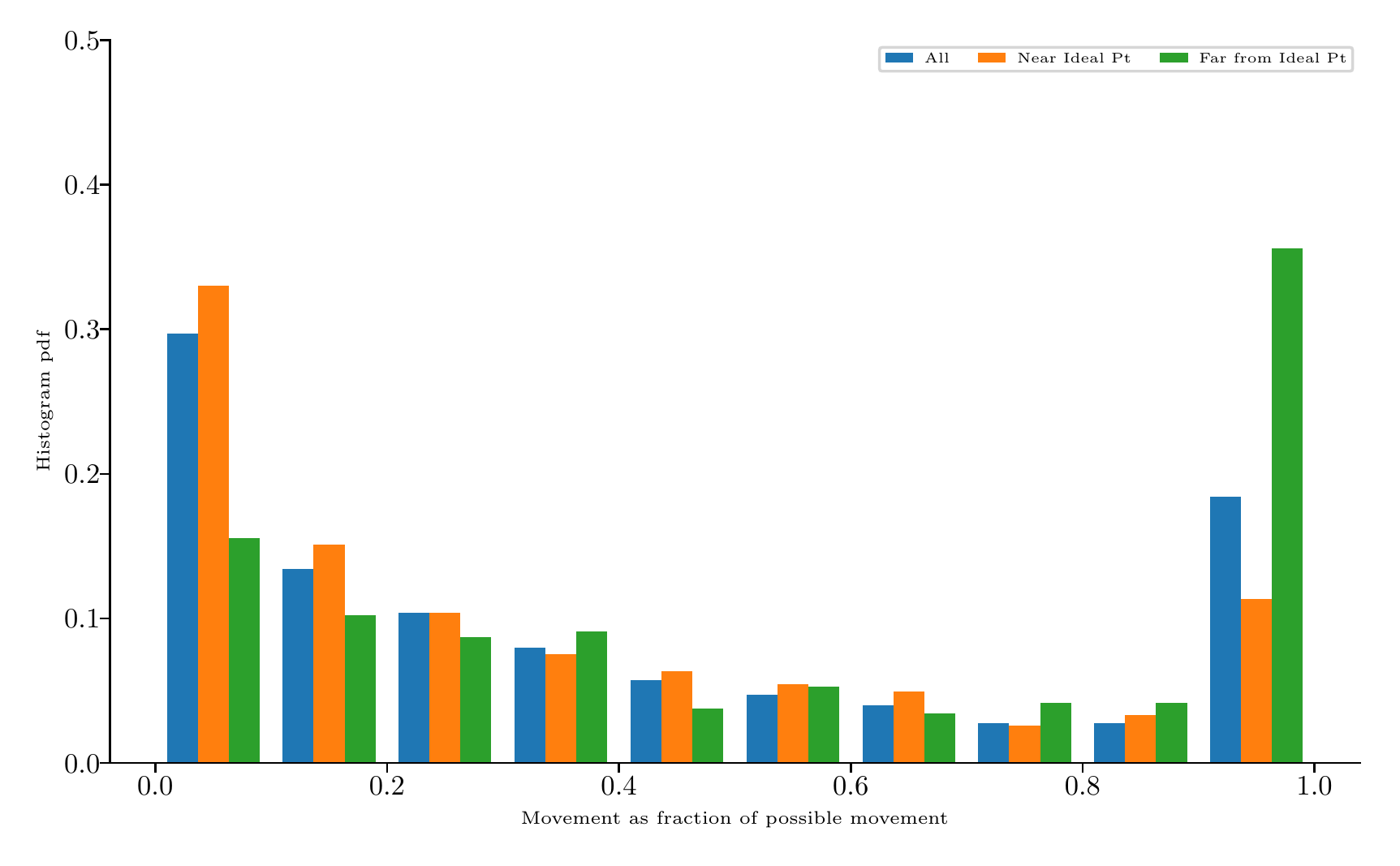}
	\caption{Fraction of possible movement in each dimension in $\mathcal{L}^\infty$, conditioned on distance to ideal pt. The `All' condition contains data from all three $\mathcal{L}^\infty$ instances, whereas the others only from the instance that also did full elicitation. This plot only includes those people who provided an explanation \textbf{as long or longer than the median} explanation provided (197 characters).}
	\label{fig:histogramofcreditslinf_cared}
\end{figure}
\begin{figure}[htb!]
	\centering
	\includegraphics[width=4.5in]{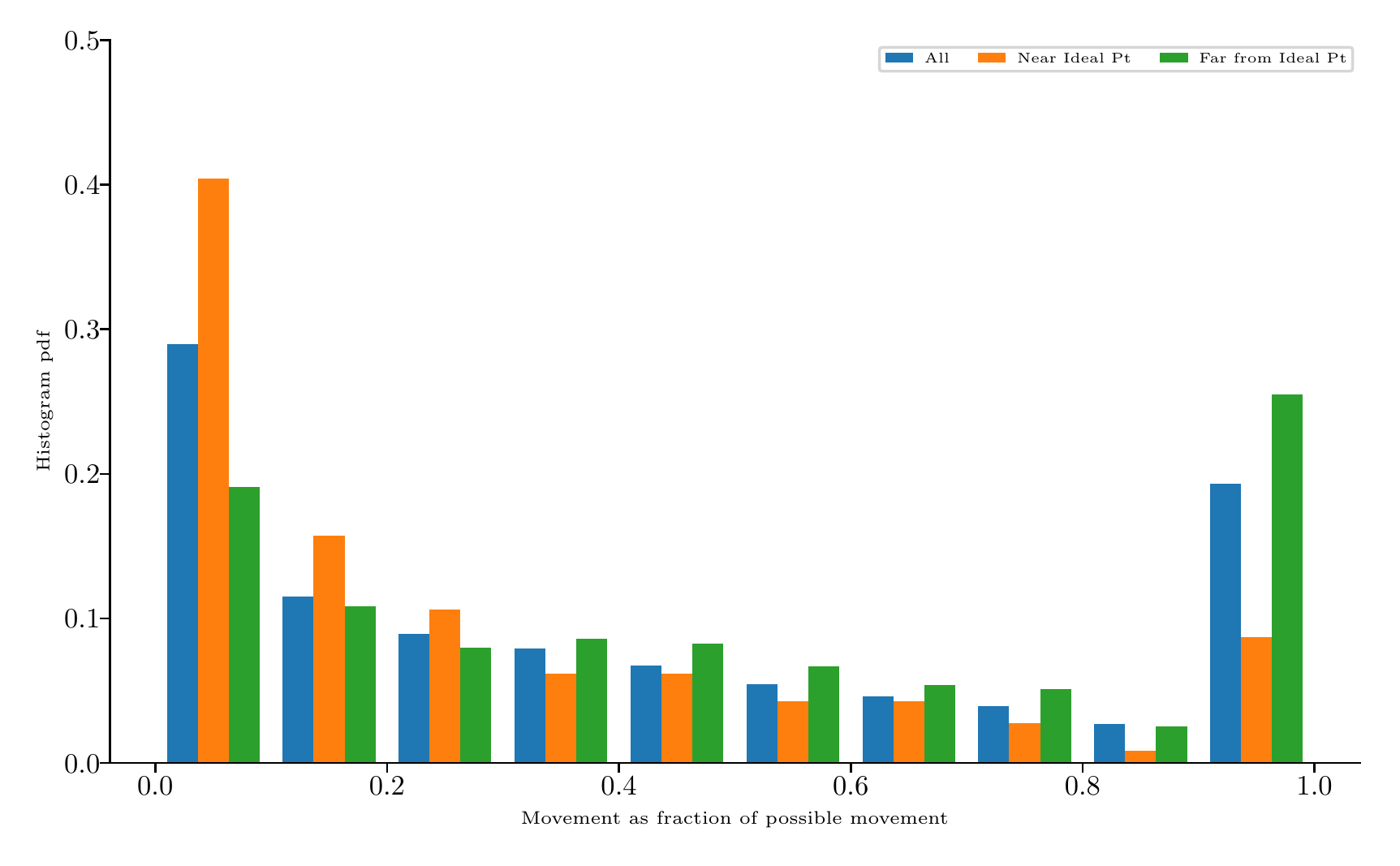}
	\caption{Fraction of possible movement in each dimension in $\mathcal{L}^\infty$, conditioned on distance to ideal pt. The `All' condition contains data from all three $\mathcal{L}^\infty$ instances, whereas the others only from the instance that also did full elicitation. This plot only includes those people who provided an explanation \textbf{shorter than the median} explanation provided (197 characters).}
	\label{fig:histogramofcreditslinf_caredNOT}
\end{figure}
\FloatBarrier
\section{Proofs}
In this appendix, we include proofs for all the theorems in the paper. 

\subsection{Known SSGM Results}
\begin{theorem} \cite{nemirovski_robust_2009,strassen_existence_1965}
	Let $\theta \in \Theta$ be a random vector with distribution P. Let $\bar{f}(x) = \textup{E}[f(x, \theta)] = \int_{\Theta}f(x, \theta) dP(\theta)$, for $x\in \mathcal{X}$, a non-empty bounded closed convex set, and assume the expectation is well-defined and finite valued. Suppose that $f(\cdot,\theta), \theta \in \Theta$ is convex and $\bar{f}(\cdot)$ is continuous and finite valued in a neighborhood of point $x$. 
	For each $\theta$, choose any $ g(x, \theta) \in \partial f(x, \theta)$. Then, there exists $\bar{g}(x) \in \partial \bar{f}(x)$ s.t. $\bar{g}(x) = \textup{E}_\theta[g(x, \theta)]$. 
	\label{thm:leib} 
\end{theorem}

This theorem says that the expected value of the sub-gradient of the utility at any point $x$ across voters is a subgradient of the societal utility at $x$, irrespective of how the voters choose the subgradient when there are multiple subgradients, i.e., when the utility function is not differentiable.   
This key result allows us to use the subgradient of utility function of a sampled voter as an unbaised estimate of the societal subgradient.

Now, consider a convex function $f$ on a non-empty bounded closed convex set $\mathcal{X} \subset \bbR^M$, and use $[\cdot]_\mathcal{X}$ to designate the projection operator. Starting with some $x_0\in \mathcal{X}$, consider the SSGM update rule $x_{t} = [x_{t-1} - r_t (\bar{g}_t + z_t + b_t)]_\mathcal{X}$, where $z_t$ is a zero-mean random variable and $b_t$ is a constant, and $\bar{g}_t \in \partial f(x_t)$. Let $\textup{E}_t[\cdot]$ be the conditional expectation given $\mathcal{F}_t$, the $\sigma$-field generated by $x_0,x_1,\dots,x_t$. Then we have the following convergence result.
\begin{theorem} \cite{jiang_scheduling_2010} Consider the above update rule.  If
	\begin{align*}
	&f(\cdot)\text{ has a unique minimizer }x^* \in \mathcal{X}\\
	&r_t > 0, \sum _t r_t = \infty, \sum _t r_t^2 < \infty\\
	&\exists C_1 \in \bbR < \infty \text{ s.t. }\|\partial f(x)\|_2 \leq C_1,\forall\, x \in \mathcal{X} \\
	&\exists C_2 \in \bbR < \infty \text{ s.t. } \textup{E}_t[\|z_t\|^2] \leq C_2, \forall\, t \\
	&\exists C_3 \in \bbR < \infty \text{ s.t. }\|b_t\|_2 \leq C_3,\forall\, t \\
	&\sum_{t} r_t \|b_t\| < \infty\text{ w.p. }1
	\end{align*}
	Then $x_t \to x^*$ w.p. 1 as $t \to \infty$.  
	\label{thm:ssgd}
\end{theorem}

Note: \citet{jiang_scheduling_2010} prove the result for gradients, though the same proof follows for subgradients. Only the inequality $[x^* - x_t]^T g_t \leq f(x^*) - f(x_t)$ for gradient $g_t$ at iteration $t$ is used, which holds for subgradients. \citet{boyd_subgradient_2006} provide a general discussion of subgradient methods, along with similar results. \citet{shor_nondifferentiable_1998}, in Theorem 46, provide a convergence proof for the stochastic subgradient method without projections and the extra noise terms.

\subsection{Mapping {\sc ILV} to SSGM}

As described in Section~\ref{sec:conv}, suppose that $h_\mathcal{X}$ is the induced probability distribution on the ideal values of the voters. In the following discussion, we will refer to voters and their ideal solutions interchangeably.

Next, we restate {\sc ILV} without the stopping condition so that it looks like the stochastic subgradient method. Consider Algorithm~\ref{alg:votermovetomin}.
\begin{algorithm}
	\SetAlgorithmName{Algorithm}{}[
	Start at some $x_0\in \mathcal{X}$. For $t\geq 1$,
	\begin{itemize}
		\item Sample voter $v_t \in \mathcal{V}$ from $h_\mathcal{V}$.
		\item  Compute $x_{t} = [x_{t-1} - r_t\tilde{g}_{v_t}(x_t)]_\mathcal{X}$, where $r_t=\frac{r_0}{t}$ and $r_t\tilde{g}_{v_t}(x_t)$ is movement given by $v_t$.
	\end{itemize}
	\caption{{\sc ILV}}
	\label{alg:votermovetomin}
\end{algorithm}

We want to minimize the societal cost, $\bar{f}(x) = \textup{E}[f_v(x)]$. From Theorem~\ref{thm:leib}, it immediately follows that if each voter $v$ articulates a subgradient of her utility function for all $x$, i.e. $\tilde{g}_{v}(x) \in \partial f_v(x)$, then from Theorem~\ref{thm:ssgd}, we can conclude that the algorithm converges. However, users may not be able to articulate such a subgradient. Instead, when the voters respond correctly to query~\eqref{query} (i.e. move to their favorite point in the given $\mathcal{L}^q$ neighborhood), we have  
\begin{align}\label{def:gtilde}
\tilde{g}_{v_t}(x_t) &= \frac{x_t - \arg\min_x[f_{v_t}(x) : \|x - x_t\|_q \leq r_t]}{r_t}. 
\end{align} 

\noindent Furthermore, for all the proofs, we assume the following.
\begin{enumerate}[leftmargin=*]
	\item The solution space $\mathcal{X}\subset \mathbb{R}^M$ is non-empty, bounded, closed, and convex.
	\item Each voter $v$ has a unique ideal solution $x_v\in \mathcal{X}$.
	\item The ideal point $x_v$ of each voter is drawn independently from a probability distribution with a bounded and measurable density function $h_\mathcal{X}$ on $M$ dimensions: $\textrm{there exists } C \text{ s.t. }\forall\, x$ we have $h_{\mathcal{X}} (x) \leq C$. This assumption allows us to bound the probability of errors that occur in small regions of the space.
\end{enumerate}
\subsection{Proof of Theorem~\ref{thm:pqmain}}

Let the disutility, or cost to voter $v \in \mathcal{V}$ be $f_v(x) = \|x - x_v\|_p$ for all $x\in \mathcal{X}$.   
\noindent We use the following technical lemma:
\begin{lemma}
	For $q \in \{1,2, \infty\}$, there exists $K_2 \in \bbR^+$ s.t. $\|\tilde{g}_{v}(x) -  g_t\|_2 \leq K_2$, $\forall\,\, g_t \in \partial f_{v}(x)$ for any  $v$ and  $x$.
	\label{lem:boundedgtildeerror}
\end{lemma}

\noindent The lemma bounds the error in the movement direction from the gradient direction, by noting that both the movement direction and the gradient direction have bounded norms. 
\\\\\noindent We also need the following lemma, which is proved separately for each case in the following sections. 
\begin{lemma}
	Suppose that $f_v(x) \triangleq \|x_v - x \|_p$ and define the function 
	$$A_t\triangleq \mathbb{I}{\{\tilde{g}_{v_t}(x_t)\notin\partial f_{v_t}(x_t)\}},$$
	where $\tilde{g}_{v_t}(x_t)$ is as defined in~\eqref{def:gtilde}. Then there exists $C\in \bbR$ s.t. $\forall\,\, n$, $\textup{P}(A_t = 1 | \mathcal{F}_t) \leq C r_t$, when $(p=2, \,q=2)$, $(p=1, \,q=\infty)$, or $(p=\infty, \,q=1)$. 
	\label{lem:whengoodcondition}
\end{lemma}
\noindent The lemma can be interpreted as follows: $A_t$ indicates a `bad' event, when a voter may not be providing a true subgradient of her utility function. However, the probability of the event occurring vanishes with $r_t$, which, as we will see below, is the right rate for the algorithm to converge.

\thmpqmain*
\begin{proof}
	We will show that Algorithm~\ref{alg:votermovetomin} meets the conditions in Theorem~\ref{thm:ssgd}. Let $b_t \triangleq \textup{E}_t[\tilde{g}_{v_t}(x_t)] - \bar{g}_t$ and $z_t \triangleq \tilde{g}_{v_t}(x_t) - \textup{E}_t[\tilde{g}_t]$, for some $\bar{g}_t\in \partial \bar{f}(x_t)$. Then, $\tilde{g}_{v_t}(x_t)$ can be written as $\tilde{g}_{v_t}(x_t) = \bar{g}_t + z_t + b_t$. We show that $b_t$, $z_t$ meet the conditions in the theorem, and so the algorithm converges.\\\\
	\noindent Let $A_t$ be the indicator function described in Lemma~\ref{lem:whengoodcondition}. Then, for some $\bar{g}_t\in \partial \bar{f}(x_t)$, 
	\begin{align*}
	b_t={}& \textup{E}_t[\tilde{g}_{v_t}(x_t)] - \bar{g}_t\\
	={}& \textup{E}_t[\tilde{g}_{v_t}(x_t)] - \textup{E}_t[g_t] & \text{Theorem}~\ref{thm:leib}, \text{ i.i.d sampling of $v$}\\
	={}& \textup{P}(A_t = 1 | \mathcal{F}_t)(\textup{E}_t[\tilde{g}_{v_t}(x_t) | A_t = 1] - \textup{E}_t[g_t| A_t = 1])  \\&+ \textup{P}(A_t = 0 | \mathcal{F}_t)(\textup{E}_t[\tilde{g}_{v_t}(x_t) | A_t = 0] - \textup{E}_t[g_t| A_t = 0])\\
	={}& \textup{P}(A_t = 1 | \mathcal{F}_t)(\textup{E}_t[\tilde{g}_{v_t}(x_t) | A_t = 1] - \textup{E}_t[g_t| A_t = 1])&\\
	\leq{}& C r_t (\textup{E}_t[\tilde{g}_{v_t}(x_t) | A_t = 1] - \textup{E}_t[g_t| A_t = 1]).&\text{Lemma}~\ref{lem:whengoodcondition}
	\end{align*}
	Combining with Lemma~\ref{lem:boundedgtildeerror}, and the fact that $r_t=r_0/t$, we have 
	$$\sum r_t \|b_t\| \leq{} \infty \text{ and there exists } C_1 \in \bbR < \infty \text{ s.t. }\|b_t\|_2 \leq C_1,\forall\, t.$$
	Finally, note that  $\|z_t\| \triangleq \|\tilde{g}_{v_t}(x_t) - \textup{E}_t[\tilde{g}_{v_t}(x_t)]\|$ is bounded for each $t$ because the $\|\tilde{g}_{v_t}(x_t)\|$ is bounded as defined. Thus, all the conditions in Theorem~\ref{thm:ssgd} are met for both $b_t$ and $z_t$, and the algorithm converges.
\end{proof}

\subsection{Proof of Theorem~\ref{thm:pqother}}
Instead of moving to their favorite point on the ball, voters now instead move in the direction of the gradient of their utility function to the boundary of the given neighborhood. In this case, we have:
\begin{equation}
\tilde{g}_{v_t}(x_t) = \frac{g_{v_t}}{\|g_{v_t}\|_q}; \textrm{ for }g_{v_t} \in \partial f_{v_t}(x_t).
\label{def:gtildepartial}
\end{equation}

The key to the proof is the following observation, that the $q$ norm of the gradient of the $p$ norm, except at the ideal points on each dimension, is constant. This observation is formalized in the following lemma:
\begin{lemma}
	$\forall\, (p, q)$ s.t. $p>0$, $q>0$, and $1/p + 1/q = 1$, $\| \nabla \|x - x_v \|_p \|_q = 1, \forall\, x$ s.t. $ x^m\neq x^m_v$ for any $m$. 
	\label{lem:constantqnorm}
\end{lemma}

\thmpqother*
\begin{proof}
	Since the probability of picking a voter $v$ such that $x_t^m = x^m_v$ for some dimension $m$ is $0$, we have $\tilde{g}_{v_t}(x_t) = g_{v_t}$ for $g_{v_t} = \nabla f_{v_t}(x_t)$. Thus we obtain the gradient exactly, and hence Theorem~\ref{thm:ssgd} applies with $b_t = 0$ for all $t$.
\end{proof}

\subsection{Proof of Propositions}
We now turn our attention to the case of Weighted Euclidean utilities and show that Algorithm~\ref{alg:votermovetomin} converges to the societal optimum. The analogue to Lemma~\ref{lem:whengoodcondition} for this case is (proved in the following subsection):
\begin{lemma}
	Suppose that $f_v(x) \triangleq \sum_{k = 1}^K \frac{w_v^k}{\|w_v\|_2}\|x^k - x^k_v\|_2$, and define the function $$A_t\triangleq \mathbb{I}{\{\tilde{g}_{v_t}(x_t)\notin \partial f_{v_t}(x_t)\}},$$ where $\tilde{g}_{v_t}(x_t)$ is as defined in~\eqref{def:gtilde} for $q=2$. Then there exists $C\in \bbR$ s.t. $\forall\, n$, $\textup{P}(A_t = 1 | \mathcal{F}_t) \leq C r_t$. 
	\label{lem:mainlemmaweightedeuc}
\end{lemma}

\thmotherone*

\begin{proof}
	The proof is then similar to that of Theorem~\ref{thm:pqmain}, and the algorithm converges to $x^* = \arg\min \textup{E}\left[\frac{\sum_{k = 1}^K w_v^k\|x^k - x^k_v\|_2}{\|w_v\|_2}\right]$. 
\end{proof}
Now, we sketch the proof for fully decomposable utility functions and $\mathcal{L}^\infty$ neighborhoods.
\thmothertwo*
\begin{proof}
	Consider each dimension separately. If $x^m_{t-1} < x^m_v$, then the sampled voter increases $x^m_{t-1}$ by $r_t$ as long as $x^m_{t-1} + r_t\leq x^m_v$. On the other hand if $x^m_{t-1} > x^m_v$, then the sampled voter decreases $x^m_{t-1}$ by $r_t$ as long as $x^m_{t-1} - r_t\geq x^m_v$. Thus except for when a voter's ideal solution is too close to the current point, the algorithm can be seen as performing SSGM on each dimension separately as if the utility function was $\mathcal{L}^1$ (the absolute value) on each dimension. Thus a proof akin to that of Theorem~\ref{thm:pqmain} with $p=1, q= \infty$ holds.
\end{proof}

\subsection{Proof of Theorem~\ref{thm:equiv}}
We now show that the algorithm finds directional equilibria in the following sense: if under a few conditions a trajectory of the algorithm converges to a point, then that point is a directional equilibrium.

\thmequiv*
\begin{proof}Suppose $x^*$ is not a directional equilibrium, i.e. $\exists \epsilon > 0$ s.t. $\|G(x^*)\|_2 = \epsilon$. Consider a $\delta$-ball around $x^*$, $B_\delta \triangleq \{x : \|x^* - x \|_2 < \delta\}$, with $\delta, \epsilon_2 > 0$ chosen such that $\exists m \in \{1 \dots M\}$ s.t. $\forall x \in B_\delta$, $\sign (G_m(x)) = \sign(G_m(x^*))$ and $|G_m(x)| > \epsilon_2$, i.e. the gradient in the $m$th dimension does not change sign and has magnitude bounded below. Such a $\delta, \epsilon_2$ exists by the continuity assumption (if $x^*$ is not a directional equilibrium, at least 1 dimension of $G(x^*)$ is non-zero and thus one can construct a ball around $x^*$ such that $G(x), x\in B_\delta$ in that dimension satisfies the conditions). 
	\\\\Now, one can show that the probability of leaving neighborhoods around $x^*$ goes to 1: $\forall t>0, 0<\delta_2<\delta$, w.p. 1 $\exists \tau \geq t$ s.t. $\|x_\tau - x^* \|_2 > \delta_2$.   
	\\Suppose $x_t \in B_{\delta_2}$ (otherwise $\tau = t$ satisfies), $r_k = \frac{1}{k}$. 
	\begin{align*}
	x_\tau &= x_t + \sum_{k=t}^{\tau} \Delta x_k & \Delta x_k \triangleq - r_k \frac{\nabla f^{v_k}(x_k)}{\| \nabla f^{v_k}(x_k)\|_2}\\
	\|x_\tau - x^* \|_2 &= \| x_t - x^* + \sum_{k=t}^{\tau} \Delta x_k \|_2 \\
	&\geq \|\sum_{k=t}^{\tau}\Delta x_k \|_2 - \| x_t - x^*\|_2 \\
	&\geq \|\sum_{k=t}^{\tau}\Delta x_k \|_2 - \delta_2 \\
	\|\sum_{k=t}^{\tau}\Delta x_k \|_2 &\geq | \sum_{k=t}^{\tau}\Delta x_{k,m} | & \text{defn of $\|\cdot\|_2$}\\
	&= |\bbE_v\left[\sum_{k=t}^{\tau}\Delta x_{k,m}\right] + \sum_{k=t}^{\tau}\Delta x_{k,m} - \bbE_v\left[\sum_{k=t}^{\tau}\Delta x_{k,m}\right] | \\
	&\geq |\bbE_v\left[\sum_{k=t}^{\tau}\Delta x_{k,m}\right] | - |\sum_{k=t}^{\tau}\Delta x_{k,m} - \bbE_v\left[\sum_{k=t}^{\tau}\Delta x_{k,m}\right] |
	\end{align*}
	By Hoeffding's inequality,
	\begin{align*}
	Pr\left(\sum_{k=t}^{\tau}\Delta x_{k,m} - \bbE_v\left[\sum_{k=t}^{\tau}\Delta x_{k,m}\right] \geq \epsilon_3 \right) &\leq \exp\left[-\frac{2(\tau - t)^2\epsilon_3^2}{2\sum_{k=t}^{\tau} \frac{1}{k}}\right] \\
	&\to 0 \text{ as } \tau \to \infty & 
	\end{align*}
	Furthermore, by the continuity assumption,
	\begin{align*}
	|\bbE_v\left[\sum_{k=t}^{\tau}\Delta x_{k,m}\right]| &\triangleq |\sum_{k=t}^{\tau}r_kG_m(x_k)| \\
	&\to \infty \text{ as } \tau \to \infty \text{ while } x_k \in B_{\delta_2}
	\end{align*}
	Thus, $Pr(\|x_\tau - x^* \|_2 > \delta_2) \to 1 $ as ${\tau \to \infty}$. Thus, if an infinite trajectory converges to $x^*$, then w.p. $1$, then $x^*$ is a directional equilibrium.
\end{proof}

\subsection{Proofs of Lemmas}
\label{sec:lem2proofs}
\textbf{Lemma~\ref{lem:boundedgtildeerror}}
For $q \in \{1,2, \infty\}$, $\exists K_2 \in \bbR^+ < \infty$ s.t. $\|\tilde{g}_{v_t} -  g_t\|_2 \leq K_2$, $\forall\, g_t \in \partial f_{v_t}(x_t), v_t, x_t$.
\begin{proof}
	\begin{align*}
	\|\tilde{g}_{v_t}(x_t) -  g_t\|_2 &\leq \|\tilde{g}_{v_t}(x_t)\|_2 + \|g_t\|_2 \\
	&= \frac{\|x_t - \arg\min_x[ \|x - x_{v_t}\|_p : \|x - x_t\|_q \leq r_t]\|_2}{r_t} + \|g_t\|_2\\
	&\leq K_1 + \|g_t\|_2\\
	&\leq K_2
	\end{align*} 
	$\text{for some }K_1,K_2 \in \bbR^+$. The second inequality follows from the fact that for finite M-dimensional vector spaces, $\|y\|_2 \leq \|y\|_1$ and $\|y\|_2 \leq \sqrt{M} \|y\|_\infty$. The third follows from the norm of the subgradients of the $p$ norm being bounded. 
\end{proof}
\noindent\textbf{Lemma~\ref{lem:whengoodcondition}}, case ($p = 2, q = 2$). 
\begin{proof}
	Remember that $A_t\triangleq \mathbb{I}{\{\tilde{g}_{v_t}(x_t)\notin\partial f_{v_t}(x_t)\}}$. 
	Let $B_t = \mathbb{I}\{\|x_{v_t} - x_t\|_2 \leq r_t\}$. We show that A) $B_t = 0 \implies A_t = 0$, and B) $\exists C\in \bbR$ s.t. $\textup{P}(B_t = 1 | \mathcal{F}_t) \leq C r_t$. Then, $\exists C\in \bbR$ s.t. $\textup{P}(A_t = 1 | \mathcal{F}_t) \leq C r_t$. \\\\
	\textbf{Part A}, $B_t = 0 \implies \tilde{g}_{v_t}(x_t) = g_t$, for some $g_t \in \partial f_{v_t}(x_t)$:\\
	First, note that
	\begin{align*}
	\partial f_{v_t}(x) &= \partial \|x - x_{v_t}\|_2 \\
	&= \begin{cases} 
	\{ \frac{x - x_{v_t}}{\|x_{v_t} - x\|_2}\} & x \neq x_{v_t} \\
	\{g : \|g\|_2 \leq 1\}  & x = x_{v_t} 
	\end{cases}
	\end{align*}
	If $\|x_{v_t} - x_t\|_2 > r_t$
	, then
	\begin{align*}
	\arg\min_x[ \|x - x_{v_t}\|_2 : \|x - x_t\|_2 \leq r_t] &= x_t + r_t \frac{x_{v_t} - x_t}{\|x_{v_t} - x_t\|_2}
	\end{align*}
	Then,
	\begin{align*}
	\tilde{g}_{v_t}(x_t) &= \frac{x_t - \arg\min_x[ \|x - x_{v_t}\|_2 : \|x - x_t\|_2 \leq r_t]}{r_t} & \text{Definition}\\
	&= \frac{x_t - \left(x_t + r_t \frac{x_{v_t} - x_t}{\|x_{v_t} - x_t\|_2}\right)}{r_t}\\
	&= \frac{x_t - x_{v_t}}{\|x_{v_t} - x_t\|_2} \\
	&\in \partial f_{v_t}(x_t)
	\end{align*} 
	\noindent\textbf{Part B}, $\exists C\in \bbR$ s.t. $\textup{P}(B_t = 1 | \mathcal{F}_t) \leq C r_t$:\\
	\begin{align*}
	\textup{P}(B_t = 1 | \mathcal{F}_t) &= \textup{P}(\|x_{v_t} - x_t\|_2 \leq r_t |\mathcal{F}_t)\\
	&= \int_{x \in \{x : \|x - x_t\|_2 \leq r_t\}} h_{\mathcal{X} | \mathcal{F}_t} (x) dx \\
	&= \int_{x \in \{x : \|x- x_t\|_2 \leq r_t\}} h_{\mathcal{X}} (x) dx & v \text{ drawn independent of history}\\
	&\leq Cr_t^2 & \text{bounded } h_{\mathcal{X}} \\
	& \leq Cr_t & r_t \leq 1 \text{ eventually}
	\end{align*}
	for some $C\in \bbR < \infty$. Note that $C$ depends on the volume of a sphere in $M$ dimensions. 
\end{proof}
\noindent\textbf{Lemma~\ref{lem:whengoodcondition}}, case ($p = 1, q = \infty$).
\begin{proof}
	Let $h_{v_t}(x_t)\triangleq \begin{bmatrix}
	\sign(x_{v_t}^1 - x_t^{1}), &
	\hdots,&
	\sign(x_{v_t}^m - x^m_t), &
	\hdots,&\
	\sign(x_{v_t}^M- x_t^{M})
	\end{bmatrix}^T$\\\\
	Let $B_t = \mathbb{I}\{\exists m, |x_{v_t}^m - x^m_t| \leq r_t\}$. We show the the same two parts as in the above proof.\\\\
	\textbf{Part A}, $B_t = 0 \implies \tilde{g}_{v_t}(x_t) = g_t$, for some $g_t \in \partial f_{v_t}(x_t)$:\\
	First, note that the subgradients are
	\begin{align*}
	\partial f_{v_t}(x) &= \partial \|x - x_{v_t}\|_1 \\
	&= \{g : \|g\|_\infty \leq 1, g^T(x - x_{v_t}) = \|x - x_{v_t}\|_1\}
	\end{align*}
	\\\\If $\forall\, m, |x_{v_t}^m - x^m_t| > r_t$, then
	\begin{align*}
	\arg\min_x[ \|x - x_{v_t}\|_1 : \|x - x_t\|_\infty \leq r_t] &= x_t + r_t h_{v_t}(x_t)
	\end{align*}
	Then,
	\begin{align*}
	\tilde{g}_{v_t}(x_t) &= \frac{x_t - \arg\min_x[ \|x - x_{v_t}\|_1 : \|x - x_t\|_\infty \leq r_t]}{r_t} & \text{Definition}\\
	&= \frac{x_t - \left(x_t + r_t h_{v_t}(x_t)\right)}{r_t}\\
	&= -h_{v_t}(x_t)\\
	&\in \partial f_{v_t}(x_t)
	\end{align*}
	\noindent\textbf{Part B}, $\exists C\in \bbR$ s.t. $\textup{P}(B_t = 1 | \mathcal{F}_t) \leq C r_t$:\\
	\begin{align*}
	\textup{P}(B_t = 1 | \mathcal{F}_t) &= \textup{P}(\exists m: |x_{v_t}^m - x^m_t| \leq r_t |\mathcal{F}_t)\\
	&= \int_{x \in \{x : \exists m, |x^m - x^m_t| \leq r_t\}} h_{\mathcal{X} | \mathcal{F}_t} (x) dx \\
	&= \int_{x \in \{x : \exists m, |x^m - x^m_t| \leq r_t\}} h_{\mathcal{X}} (x) dx & v \text{ drawn independent of history}\\
	&\leq Cr_t & \text{bounded }h_\mathcal{X}, \text{ fixed } M, \text{ bounded }\mathcal{X}
	\end{align*}
	for some $C\in \bbR < \infty$. In the last line, $C  \approxeq 2M (\text{diameter}(\mathcal{X}))$, based on the volume of the slices around the ideal points on each dimension.  
\end{proof}

\noindent\textbf{Lemma~\ref{lem:whengoodcondition}}, case ($p = 1, q = \infty$). 
\begin{proof}
	Let $ {\bar{m}_t} \in \arg \max_m |x_{v_t}^m - x^m_t|$,\\
	Let $h_{v_t}(x_t) \triangleq \begin{bmatrix}
	0,&0, &
	\hdots,&0,&
	\sign(x^{\bar{m}_t}_t - x_{v_t}^{\bar{m}_t}),
	&0,\hdots,&0, &
	0
	\end{bmatrix}^T$,\\ Let $B_t \triangleq \mathbb{I}\{\exists m\neq
	\bar{m}_t : |x_{v_t}^{\bar{m}_t} - x^{\bar{m}_t}_t| < |x_{v_t}^{m} - x^m_t| + r_t\}$. We show the the same two parts as in the above proofs.\\\\
	\textbf{Part A}, $B_t = 0 \implies \tilde{g}_{v_t}(x_t) = g_t$, for some $g_t \in \partial f_{v_t}(x_t)$:\\
	
	First, note that when $B_t = 0$, the set of subgradients is
	\begin{align*}
	\partial f_{v_t}(x) &= \partial \|x - x_{v_t}\|_\infty \\
	&= \{h_{v_t}(x_t)\}
	\end{align*}
	Also when $B_t = 0$,
	\begin{align*}
	\arg\min_x[ \|x - x_{v_t}\|_\infty : \|x - x_t\|_1 \leq r_t] &= x_t - r_t h_{v_t}(x_t)
	\end{align*}
	Then,
	\begin{align*}
	\tilde{g}_{v_t}(x_t) &= \frac{x_t - \arg\min_x[ \|x - x_{v_t}\|_1 : \|x - x_t\|_\infty \leq r_t]}{r_t} & \text{Definition}\\
	&= \frac{x_t - \left(x_t - r_t h_{v_t}(x_t)\right)}{r_t}\\
	&= h_{v_t}(x_t)\\
	&\in \partial f_{v_t}(x_t)
	\end{align*} 
	\noindent\textbf{Part B}, $\exists C\in \bbR$ s.t. $\textup{P}(B_t = 1 | \mathcal{F}_t) \leq C r_t$:\\
	\begin{align*}
	\textup{P}(B_t = 1 | \mathcal{F}_t) &= \textup{P}(\mathbb{I}\{\exists m\neq
	\bar{m}_t : |x_{v_t}^{\bar{m}_t} - x^{\bar{m}_t}_t| < |x_{v_t}^{m} - x^m_t| + r_t\} |\mathcal{F}_t)\\
	&= \int_{x \in \{x : \exists m \neq
		\bar{m}_t\text{ s.t. } |x_{v_t}^{\bar{m}_t} - x^{\bar{m}_t}_t| < |x_{v_t}^{m} - x^m_t| + r_t\}} h_{\mathcal{X} | \mathcal{F}_t} (x) dx \\
	&= \int_{x \in \{x : \exists m \neq
		\bar{m}_t\text{ s.t. } |x_{v_t}^{\bar{m}_t} - x^{\bar{m}_t}_t| < |x_{v_t}^{m} - x^m_t| + r_t\}} h_{\mathcal{X}} (x) dx & v \text{ drawn independent of history}\\
	&\leq Cr_t & \text{bounded } h_{\mathcal{X}}, \text{ fixed } M, \text{ bounded }\mathcal{X}
	\end{align*}
	for some $C\in \bbR < \infty$. Note that $C \approxeq 2M^2 (\text{diameter}(\mathcal{X}))$, based on the volume of the slices around each dimension.   
\end{proof}
\noindent \textbf{Lemma~\ref{lem:constantqnorm}}
$\forall\, (p, q)$ s.t. $p>0, q>0, 1/p + 1/q = 1$, $\| \nabla \|x - x_v \|_p \|_q = 1, \forall\, x$ s.t. $\forall\, m, x^m\neq x^m_v$. 
\begin{proof}
	If $x^m \neq x^m_v, \forall\, m:$
	\begin{align*}
	\nabla_m \|x - x_v \|_p &= \nabla_m \left(\sum_m |x^m - x^m_v |^p\right)^{1/p}\\
	&= \frac{1}{p} \frac{\nabla_m |x^m - x^m_v |^{p}}{\left(\sum_m |x^m - x^m_v |^p\right)^{1 - 1/p}} \\
	&= \frac{|x^m - x^m_v |^{p-1} \left(\nabla_m |x^m - x^m_v |\right)}{\|x - x_v \|_p^{p - 1}}
	\end{align*}
	Then
	\begin{align*}
	\|\nabla \|x - x_v \|_p\|_q &= \|\frac{|x^m - x^m_v |^{p-1} \left(\nabla_m |x^m - x^m_v |\right)}{\|x - x_v \|_p^{p - 1}}\|_q\\
	&= \frac{1}{{\|x - x_v \|_p^{p - 1}}} \left(\sum_m \left||x^m - x^m_v |^{p-1} \left(\nabla_m |x^m - x^m_v |\right)\right|^q\right)^{1/q} \\
	&= \frac{1}{{\|x - x_v \|_p^{p - 1}}} \left(\sum_m |x^m - x^m_v |^{(p-1)q}\right)^{1/q}\\
	&= \frac{1}{{\|x - x_v \|_p^{p - 1}}} {\|x - x_v \|_p^{p/q}} & (p-1)q = p\\
	&= 1
	\end{align*}
\end{proof}
\noindent \textbf{Lemma~\ref{lem:mainlemmaweightedeuc}} \textit{Suppose that $f_v(x) \triangleq \sum_{k = 1}^K \frac{w_v^k}{\|w_v\|_2}\|x^k - x^k_v\|_2$, and define the function $$A_t\triangleq \mathbb{I}{\{\tilde{g}_{v_t}(x_t)\notin \partial f_{v_t}(x_t)\}},$$ where $\tilde{g}_{v_t}(x_t)$ is as defined in~\eqref{def:gtilde} for $q=2$. Then there exists $C\in \bbR$ s.t. $\forall\, n$, $\textup{P}(A_t = 1 | \mathcal{F}_t) \leq C r_t$. }
\begin{proof}
	Let $B_t = \mathbb{I}\{ \exists k \text{ s.t. }\|x^k_{v_t} - x^k_t\|_2 \leq r_t\}$. We show the same two parts for $B_t$ as for the proofs for Lemma~\ref{lem:whengoodcondition}.\\
	\textbf{Part A}, $B_t = 0 \implies \tilde{g}_{v_t}(x_t) = g_t$, for some $g_t \in \partial {f_{v_t}(x_t)}$:\\
	First, note that, when $B_t = 0$,
	\begin{align*}
	\partial_m {f_{v_t}(x_t)} &= \partial_m \sum_{k = 1}^K \frac{w_v^k}{\|w_v\|_2}\|x^k - x^k_v\|_2 \\
	&= \frac{w^{k_m}}{\|w_v\|_2}
	\frac{x^m- x^m_{v_t}}{\|x^{k_m}_{v_t} - x^{k_m}_t\|_2} & \text{ where $k_m$ is the subspace that contains the $m$th dimension}
	\end{align*}
	Also if $B_t = 0$, then
	\begin{align*}
	\arg\min_x\left[ \sum_{k = 1}^K \frac{w^{k}}{\|w_v\|_2}\|x^k - x^k_v\|_2 : \|x - x_t\|_2 \leq r_t\right] &= x_t + r_t \left[\dots,\frac{w^{k_m}}{\|w_v\|_2}
	\frac{x^m_{v_t} - x^m}{\|x^{k_m}_{v_t} - x^{k_m}_t\|_2},\dots\right]
	\end{align*}
	Then,
	\begin{align*}
	\tilde{g}_{v_t}(x_t) &= \frac{x_t - \arg\min_x[ \|x - x_{v_t}\|_2 : \|x - x_t\|_2 \leq r_t]}{r_t} & \text{Definition}\\
	&\in \partial f_{v_t}(x_t)
	\end{align*} 
	\noindent\textbf{Part B}, $\exists C\in \bbR$ s.t. $\textup{P}(B_t = 1 | \mathcal{F}_t) \leq C r_t$:\\
	\begin{align*}
	\textup{P}(B_t = 1 | \mathcal{F}_t) &= \textup{P}(\|x_{v_t} - x_t\|_2 \leq r_t |\mathcal{F}_t)\\
	&= \int_{x \in \{x : \exists k \text{ s.t. }\|x^k_{v_t} - x^k_t\|_2 \leq r_t\}} h_{\mathcal{X} | \mathcal{F}_t} (x) dx \\
	&= \int_{x \in \{x :\exists k \text{ s.t. }\|x^k_{v_t} - x^k_t\|_2 \leq r_t\}} h_{\mathcal{X}} (x) dx & v \text{ drawn independent of history}\\
	&\leq Cr_t^2 & \text{bounded } h_{\mathcal{X}} \\
	& \leq Cr_t & r_t \leq 1 \text{ eventually}
	\end{align*}
	for some $C\in \bbR < \infty$. Note that $C$ depends on $K$ and $M$. 
\end{proof}

\vskip 0.2in
\bibliography{bibliography2}

\begin{thebibliography}{34}
\providecommand{\natexlab}[1]{#1}
\providecommand{\url}[1]{\texttt{#1}}
\expandafter\ifx\csname urlstyle\endcsname\relax
  \providecommand{\doi}[1]{doi: #1}\else
  \providecommand{\doi}{doi: \begingroup \urlstyle{rm}\Url}\fi

\bibitem[Airiau and Endriss(2009)]{airiau_iterated_2009}
Stéphane Airiau and Ulle Endriss.
\newblock Iterated {Majority} {Voting}.
\newblock In \emph{Proceedings of the 1st {International} {Conference} on
  {Algorithmic} {Decision} {Theory}}, {ADT} '09, pages 38--49, Berlin,
  Heidelberg, 2009. Springer-Verlag.
\newblock ISBN 978-3-642-04427-4.
\newblock \doi{10.1007/978-3-642-04428-1_4}.
\newblock URL \url{http://dx.doi.org/10.1007/978-3-642-04428-1_4}.

\bibitem[Benjamin et~al.(2017)Benjamin, Heffetz, Kimball, and
  Lougee]{benjamin_relationship_2017}
Daniel Benjamin, Ori Heffetz, Miles Kimball, and Derek Lougee.
\newblock The relationship between the normalized gradient addition mechanism
  and quadratic voting.
\newblock \emph{Public Choice}, 172\penalty0 (1):\penalty0 233--263, July 2017.
\newblock ISSN 1573-7101.
\newblock \doi{10.1007/s11127-017-0414-3}.
\newblock URL \url{https://doi.org/10.1007/s11127-017-0414-3}.

\bibitem[Benjamin et~al.(2013)Benjamin, Heffetz, Kimball, and
  Szembrot]{benjamin_aggregating_2013}
Daniel~J. Benjamin, Ori Heffetz, Miles~S. Kimball, and Nichole Szembrot.
\newblock Aggregating {Local} {Preferences} to {Guide} {Marginal} {Policy}
  {Adjustments}.
\newblock \emph{The American economic review}, 103\penalty0 (3):\penalty0
  605--610, May 2013.
\newblock ISSN 0002-8282.
\newblock \doi{10.1257/aer.103.3.605}.
\newblock URL \url{https://www.ncbi.nlm.nih.gov/pmc/articles/PMC3760035/}.

\bibitem[Boyd and Mutapcic(2006)]{boyd_subgradient_2006}
Stephen Boyd and Almir Mutapcic.
\newblock Subgradient methods.
\newblock \emph{Lecture notes of EE364b, Stanford University, Winter Quarter},
  2007, 2006.

\bibitem[Cabannes(2004)]{cabannes_participatory_2004}
Yves Cabannes.
\newblock Participatory budgeting: a significant contribution to participatory
  democracy.
\newblock \emph{Environment and Urbanization}, 16\penalty0 (1):\penalty0
  27--46, April 2004.
\newblock ISSN 0956-2478, 1746-0301.
\newblock \doi{10.1177/095624780401600104}.
\newblock URL \url{http://eau.sagepub.com/content/16/1/27}.

\bibitem[Caragiannis and Procaccia(2011)]{caragiannis_voting_2011}
Ioannis Caragiannis and Ariel~D. Procaccia.
\newblock Voting almost maximizes social welfare despite limited communication.
\newblock \emph{Artificial Intelligence}, 175\penalty0 (9):\penalty0
  1655--1671, June 2011.
\newblock ISSN 0004-3702.
\newblock \doi{10.1016/j.artint.2011.03.005}.
\newblock URL
  \url{http://www.sciencedirect.com/science/article/pii/S0004370211000506}.

\bibitem[Caragiannis et~al.(2017)Caragiannis, Nath, Procaccia, and
  Shah]{caragiannis_subset_2017}
Ioannis Caragiannis, Swaprava Nath, Ariel~D. Procaccia, and Nisarg Shah.
\newblock Subset selection via implicit utilitarian voting.
\newblock \emph{Journal of Artificial Intelligence Research}, 58:\penalty0
  123--152, 2017.
\newblock URL \url{https://www.jair.org/media/5282/live-5282-9726-jair.pdf}.

\bibitem[Cheng and Zhou(2015)]{cheng_survey_2015}
Yukun Cheng and Sanming Zhou.
\newblock A {Survey} on {Approximation} {Mechanism} {Design} {Without} {Money}
  for {Facility} {Games}.
\newblock In David Gao, Ning Ruan, and Wenxun Xing, editors, \emph{Advances in
  {Global} {Optimization}}, number~95 in Springer {Proceedings} in
  {Mathematics} \& {Statistics}, pages 117--128. Springer International
  Publishing, 2015.
\newblock ISBN 978-3-319-08376-6 978-3-319-08377-3.
\newblock URL
  \url{http://link.springer.com/chapter/10.1007/978-3-319-08377-3_13}.
\newblock DOI: 10.1007/978-3-319-08377-3\_13.

\bibitem[Chung and Duggan(2018)]{chung_directional_2018}
Hun Chung and John Duggan.
\newblock Directional equilibria.
\newblock \emph{Journal of Theoretical Politics}, 30\penalty0 (3):\penalty0
  272--305, July 2018.
\newblock ISSN 0951-6298.
\newblock \doi{10.1177/0951629818775515}.
\newblock URL \url{https://doi.org/10.1177/0951629818775515}.

\bibitem[Duchi et~al.(2012)Duchi, Agarwal, and Wainwright]{duchi_dual_2012}
J.~C. Duchi, A.~Agarwal, and M.~J. Wainwright.
\newblock Dual {Averaging} for {Distributed} {Optimization}: {Convergence}
  {Analysis} and {Network} {Scaling}.
\newblock \emph{IEEE Transactions on Automatic Control}, 57\penalty0
  (3):\penalty0 592--606, March 2012.
\newblock ISSN 0018-9286.
\newblock \doi{10.1109/TAC.2011.2161027}.

\bibitem[Duchi et~al.(2015)Duchi, Jordan, Wainwright, and
  Wibisono]{duchi_optimal_2015}
J.~C. Duchi, M.~I. Jordan, M.~J. Wainwright, and A.~Wibisono.
\newblock Optimal {Rates} for {Zero}-{Order} {Convex} {Optimization}: {The}
  {Power} of {Two} {Function} {Evaluations}.
\newblock \emph{IEEE Transactions on Information Theory}, 61\penalty0
  (5):\penalty0 2788--2806, May 2015.
\newblock ISSN 0018-9448.
\newblock \doi{10.1109/TIT.2015.2409256}.

\bibitem[Flaxman et~al.(2005)Flaxman, Kalai, and McMahan]{flaxman_online_2005}
Abraham~D. Flaxman, Adam~Tauman Kalai, and H.~Brendan McMahan.
\newblock Online {Convex} {Optimization} in the {Bandit} {Setting}: {Gradient}
  {Descent} {Without} a {Gradient}.
\newblock In \emph{Proceedings of the {Sixteenth} {Annual} {ACM}-{SIAM}
  {Symposium} on {Discrete} {Algorithms}}, {SODA} '05, pages 385--394,
  Philadelphia, PA, USA, 2005. Society for Industrial and Applied Mathematics.
\newblock ISBN 978-0-89871-585-9.
\newblock URL \url{http://dl.acm.org/citation.cfm?id=1070432.1070486}.

\bibitem[Garg et~al.(2017)Garg, Kamble, Goel, Marn, and
  Munagala]{garg_collaborative_2017}
Nikhil Garg, Vijay Kamble, Ashish Goel, David Marn, and Kamesh Munagala.
\newblock Collaborative {Optimization} for {Collective} {Decision}-making in
  {Continuous} {Spaces}.
\newblock In \emph{Proceedings of the 26th {International} {Conference} on
  {World} {Wide} {Web}}, {WWW} '17, pages 617--626, Republic and Canton of
  Geneva, Switzerland, 2017. International World Wide Web Conferences Steering
  Committee.
\newblock ISBN 978-1-4503-4913-0.
\newblock \doi{10.1145/3038912.3052690}.
\newblock URL \url{https://doi.org/10.1145/3038912.3052690}.

\bibitem[Gilman(2012)]{gilman_transformative_2012}
Hollie~Russon Gilman.
\newblock Transformative {Deliberations}: {Participatory} {Budgeting} in the
  {United} {States}.
\newblock \emph{Journal of Public Deliberation}, 8\penalty0 (2), 2012.
\newblock URL
  \url{http://search.proquest.com/openview/4c3b7fd3ae888f483cacd031a749131e/1?pq-origsite=gscholar}.

\bibitem[Goel et~al.(2016)Goel, Krishnaswamy, Sakshuwong, and
  Aitamurto]{goel_knapsack_2016}
Ashish Goel, Anilesh~K. Krishnaswamy, Sukolsak Sakshuwong, and Tanja Aitamurto.
\newblock Knapsack {Voting}: {Voting} mechanisms for {Participatory}
  {Budgeting}.
\newblock 2016.
\newblock URL
  \url{http://web.stanford.edu/~anilesh/publications/knapsack_voting_full.pdf}.

\bibitem[Goel et~al.(2017)Goel, Krishnaswamy, and Munagala]{goel_metric_2017}
Ashish Goel, Anilesh~K. Krishnaswamy, and Kamesh Munagala.
\newblock Metric {Distortion} of {Social} {Choice} {Rules}: {Lower} {Bounds}
  and {Fairness} {Properties}.
\newblock In \emph{Proceedings of the 2017 {ACM} {Conference} on {Economics}
  and {Computation}}, {EC} '17, pages 287--304, New York, NY, USA, 2017. ACM.
\newblock ISBN 978-1-4503-4527-9.
\newblock \doi{10.1145/3033274.3085138}.
\newblock URL \url{http://doi.acm.org/10.1145/3033274.3085138}.

\bibitem[Hylland and Zeckhauser(1980)]{hylland_mechanism_1980}
Aanund Hylland and Richard Zeckhauser.
\newblock \emph{A mechanism for selecting public goods when preferences must be
  elicited}.
\newblock 1980.

\bibitem[Jamieson et~al.(2012)Jamieson, Nowak, and Recht]{jamieson_query_2012}
Kevin~G. Jamieson, Robert Nowak, and Ben Recht.
\newblock Query complexity of derivative-free optimization.
\newblock In \emph{Advances in {Neural} {Information} {Processing} {Systems}},
  pages 2672--2680, 2012.
\newblock URL
  \url{http://papers.nips.cc/paper/4509-query-complexity-of-derivative-free-optimization}.

\bibitem[Jiang and Walrand(2010)]{jiang_scheduling_2010}
Libin Jiang and Jean Walrand.
\newblock Scheduling and {Congestion} {Control} for {Wireless} and {Processing}
  {Networks}.
\newblock \emph{Synthesis Lectures on Communication Networks}, 3\penalty0
  (1):\penalty0 1--156, January 2010.
\newblock ISSN 1935-4185, 1935-4193.
\newblock \doi{10.2200/S00270ED1V01Y201008CNT006}.
\newblock URL
  \url{http://www.morganclaypool.com/doi/abs/10.2200/S00270ED1V01Y201008CNT006}.

\bibitem[Lalley and Weyl(2015)]{lalley_quadratic_2015}
Steven~P. Lalley and E.~Glen Weyl.
\newblock Quadratic {Voting}.
\newblock {SSRN} {Scholarly} {Paper} ID 2003531, Social Science Research
  Network, Rochester, NY, December 2015.
\newblock URL \url{http://papers.ssrn.com/abstract=2003531}.

\bibitem[Lee et~al.(2014)Lee, Goel, Aitamurto, and
  Landemore]{lee_crowdsourcing_2014}
David~Timothy Lee, Ashish Goel, Tanja Aitamurto, and Helene Landemore.
\newblock Crowdsourcing for {Participatory} {Democracies}: {Efficient}
  {Elicitation} of {Social} {Choice} {Functions}.
\newblock In \emph{Second {AAAI} {Conference} on {Human} {Computation} and
  {Crowdsourcing}}, September 2014.
\newblock URL
  \url{https://www.aaai.org/ocs/index.php/HCOMP/HCOMP14/paper/view/8952}.

\bibitem[McDermott(2010)]{mcdermott_building_2010}
Patrice McDermott.
\newblock Building open government.
\newblock \emph{Government Information Quarterly}, 27\penalty0 (4):\penalty0
  401--413, October 2010.
\newblock ISSN 0740-624X.
\newblock \doi{10.1016/j.giq.2010.07.002}.
\newblock URL
  \url{http://www.sciencedirect.com/science/article/pii/S0740624X10000663}.

\bibitem[Meir et~al.(2010)Meir, Polukarov, Rosenschein, and
  Jennings]{meir_convergence_2010}
Reshef Meir, Maria Polukarov, Jeffrey~S. Rosenschein, and Nicholas~R. Jennings.
\newblock Convergence to {Equilibria} in {Plurality} {Voting}.
\newblock In \emph{{AAAI}}, volume~10, pages 823--828, 2010.

\bibitem[Moulin(1980)]{moulin_strategy-proofness_1980}
H.~Moulin.
\newblock On {Strategy}-{Proofness} and {Single} {Peakedness}.
\newblock \emph{Public Choice}, 35\penalty0 (4):\penalty0 437--455, 1980.
\newblock ISSN 0048-5829.
\newblock URL \url{http://www.jstor.org/stable/30023824}.

\bibitem[Nemirovski et~al.(2009)Nemirovski, Juditsky, Lan, and
  Shapiro]{nemirovski_robust_2009}
Arkadi Nemirovski, Anatoli Juditsky, Guanghui Lan, and Alexander Shapiro.
\newblock Robust stochastic approximation approach to stochastic programming.
\newblock \emph{SIAM Journal on optimization}, 19\penalty0 (4):\penalty0
  1574--1609, 2009.

\bibitem[Procaccia and Rosenschein(2006)]{procaccia_distortion_2006}
Ariel~D. Procaccia and Jeffrey~S. Rosenschein.
\newblock The distortion of cardinal preferences in voting.
\newblock In \emph{International {Workshop} on {Cooperative} {Information}
  {Agents}}, pages 317--331. Springer, 2006.
\newblock URL \url{http://link.springer.com/chapter/10.1007/11839354_23}.

\bibitem[Procaccia and Tennenholtz(2009)]{procaccia_approximate_2009}
Ariel~D. Procaccia and Moshe Tennenholtz.
\newblock Approximate {Mechanism} {Design} {Without} {Money}.
\newblock In \emph{Proceedings of the 10th {ACM} {Conference} on {Electronic}
  {Commerce}}, {EC} '09, pages 177--186, New York, NY, USA, 2009. ACM.
\newblock ISBN 978-1-60558-458-4.
\newblock \doi{10.1145/1566374.1566401}.
\newblock URL \url{http://doi.acm.org/10.1145/1566374.1566401}.

\bibitem[Quarfoot et~al.(2017)Quarfoot, von Kohorn, Slavin, Sutherland,
  Goldstein, and Konar]{quarfoot_quadratic_2017}
David Quarfoot, Douglas von Kohorn, Kevin Slavin, Rory Sutherland, David
  Goldstein, and Ellen Konar.
\newblock Quadratic voting in the wild: real people, real votes.
\newblock \emph{Public Choice}, 172\penalty0 (1):\penalty0 283--303, July 2017.
\newblock ISSN 1573-7101.
\newblock \doi{10.1007/s11127-017-0416-1}.
\newblock URL \url{https://doi.org/10.1007/s11127-017-0416-1}.

\bibitem[Robbins and Monro(1951)]{robbins_stochastic_1951}
Herbert Robbins and Sutton Monro.
\newblock A {Stochastic} {Approximation} {Method}.
\newblock \emph{The Annals of Mathematical Statistics}, 22\penalty0
  (3):\penalty0 400--407, September 1951.
\newblock ISSN 0003-4851, 2168-8990.
\newblock \doi{10.1214/aoms/1177729586}.
\newblock URL \url{http://projecteuclid.org/euclid.aoms/1177729586}.

\bibitem[Shah(2007)]{shah_participatory_2007}
Anwar Shah, editor.
\newblock \emph{Participatory budgeting}.
\newblock Public sector governance and accountability series. World Bank,
  Washington, D.C, 2007.
\newblock ISBN 978-0-8213-6923-4 978-0-8213-6924-1.
\newblock OCLC: ocm71947871.

\bibitem[Shor(1998)]{shor_nondifferentiable_1998}
Naum~Z. Shor.
\newblock \emph{Nondifferentiable {Optimization} and {Polynomial} {Problems}},
  volume~24 of \emph{Nonconvex {Optimization} and {Its} {Applications}}.
\newblock Springer US, Boston, MA, 1998.
\newblock ISBN 978-1-4419-4792-5 978-1-4757-6015-6.
\newblock \doi{10.1007/978-1-4757-6015-6}.
\newblock URL \url{http://link.springer.com/10.1007/978-1-4757-6015-6}.

\bibitem[Sintomer et~al.(2008)Sintomer, Herzberg, and
  Rocke]{sintomer_porto_2008}
Yves Sintomer, Carsten Herzberg, and Anja Rocke.
\newblock From {Porto} {Alegre} to {Europe}: potentials and limitations of
  participatory budgeting.
\newblock \emph{International Journal of Urban and Regional Research},
  32\penalty0 (1):\penalty0 164--178, 2008.
\newblock URL
  \url{http://www.cpa.zju.edu.cn/participatory_budgeting_conference/english_articles/paper2.pdf}.

\bibitem[Strassen(1965)]{strassen_existence_1965}
Volker Strassen.
\newblock The {Existence} of {Probability} {Measures} with {Given} {Marginals}.
\newblock \emph{The Annals of Mathematical Statistics}, 36\penalty0
  (2):\penalty0 423--439, April 1965.
\newblock ISSN 0003-4851, 2168-8990.
\newblock \doi{10.1214/aoms/1177700153}.
\newblock URL \url{https://projecteuclid.org/euclid.aoms/1177700153}.

\bibitem[Tideman and Plassmann(2016)]{tideman_efficient_2016}
T.~Nicolaus Tideman and Florenz Plassmann.
\newblock Efficient {Collective} {Decision}-{Making}, {Marginal} {Cost}
  {Pricing}, {And} {Quadratic} {Voting}.
\newblock {SSRN} {Scholarly} {Paper} ID 2836610, Social Science Research
  Network, Rochester, NY, August 2016.
\newblock URL \url{https://papers.ssrn.com/abstract=2836610}.

\end{thebibliography}
\bibliographystyle{plainnat}

\end{document}